\documentclass{lmcs} 
\pdfoutput=1

\usepackage{lastpage}
\lmcsdoi{17}{2}{14}
\lmcsheading{}{\pageref{LastPage}}{}{}%
{Nov.~11,~2019}{May~11,~2021}{}

\ACMCCS{[{\bf Theory of computation}]: Formal languages and automata theory}

\usepackage{amsmath,amsfonts,amssymb,amsthm}

\usepackage{tikz}
  \usetikzlibrary{arrows,automata,shapes}
  \usetikzlibrary{decorations.pathmorphing}
  \usetikzlibrary{backgrounds}
  \tikzset{snake arrow/.style= {->, decorate, decoration={snake,amplitude=.4mm,segment length=1.7mm,post length=2mm}}}
  \tikzstyle{accepting}=[double distance=2pt,outer sep=1pt+\pgflinewidth]
\usepackage{microtype}
\usepackage{xspace}
\usepackage{mathrsfs}
\usepackage[textwidth=1.9cm,color=green!10,textsize=footnotesize]{todonotes}
\usepackage[english]{babel}
\usepackage{booktabs}
  \newcommand{\ra}[1]{\renewcommand{\arraystretch}{#1}}
\usepackage[section]{placeins}

\allowdisplaybreaks

\newcommand{\eps}{\varepsilon}
\newcommand{\epsnostate}{\bullet}

\DeclareMathOperator{\sub}{\rm sub}
\DeclareMathOperator{\depth}{\rm depth}

\DeclareMathOperator{\A}{\mathcal A}
\DeclareMathOperator{\B}{\mathcal B}
\DeclareMathOperator{\D}{\mathcal D}

\DeclareMathOperator{\M}{\mathcal M}

\newcommand{\R}{\mathcal{R}}

\newcommand{\tuple}[1]{\langle{#1}\rangle}
\newcommand{\blank}{\llcorner\!\lrcorner}
\newcommand{\Deltaplus}{\Delta_{\#\$}}
\DeclareMathOperator{\enc}{\rm enc}

\newcommand{\complclass}[1]{{\sc #1}\xspace}

\newcommand{\LogSpace}{\complclass{L}}
\newcommand{\NL}{\complclass{NL}}

\newcommand{\coNP}{\complclass{coNP}}
\newcommand{\PSpace}{\complclass{PSpace}}
\newcommand{\coNL}{\complclass{coNL}}
\newcommand{\Ptime}{\complclass{P}}

\newcommand{\NP}{\complclass{NP}}

\theoremstyle{plain}
\newtheorem{claim}[thm]{Claim}

\def\ie{{\em i.e.}}
\def\eg{{\em e.g.}}
\def\cf{{\em cf.}}

\begin{document}

\title{Partially Ordered Automata and Piecewise Testability}
\titlecomment{This paper is a revised and full version of our recent work partially presented at MFCS 2016~\cite{ptnfas} and SOFSEM 2018~\cite{MasopustK18}. Except for full proofs, it contains new results and provides an overview of the complexity results for several types of partially ordered NFAs and operations of universality, inclusion, equivalence, and ($k$-)piecewise testability.}

\author[T.~Masopust]{Tom\'{a}\v{s} Masopust\rsuper{a}}
\address{\lsuper{a}Department of Computer Science, Palacky University in Olomouc, and Institute of Mathematics of the Czech Academy of Sciences, Prague, Czechia}
\email{tomas.masopust@upol.cz}

\author[M. Kr\"otzsch]{Markus Kr\"otzsch\rsuper{b}}
\address{\lsuper{b}Knowledge-Based Systems Group, TU Dresden, Germany}
\email{markus.kroetzsch@tu-dresden.de}

\thanks{Supported in part by the German Research Foundation (DFG) in project number 389792660 (TRR~248, \href{https://www.perspicuous-computing.science/}{Center for Perspicuous Systems}) and Emmy Noether grant KR 4381/1-1 (DIAMOND), by the Ministry of Education, Youth and Sports under the INTER-EXCELLENCE project LTAUSA19098, by the Czech Science Foundation grant GC19-06175J, and by RVO~67985840.}

%% required for running head on odd and even pages, use suitable
%% abbreviations in case of long titles and many authors:

%%%%%%%%%%%%%%%%%%%%%%%%%%%%%%%%%%%%%%%%%%%%%%%%%%%%%%%%%%%%%%%%%%%%%%%%%%%

%% the abstract has to PRECEDE the command \maketitle:
%% be sure not to issue the \maketitle command twice!

\begin{abstract}
  \emph{Partially ordered automata\/} are automata where the transition relation induces a partial order on states. The expressive power of partially ordered automata is closely related to the expressivity of fragments of first-order logic on finite words or, equivalently, to the language classes of the levels of the Straubing-Th\'erien hierarchy. Several fragments (levels) have been intensively investigated under various names. For instance, the fragment of first-order formulae with a single existential block of quantifiers in prenex normal form is known as \emph{piecewise testable languages} or \emph{$J$-trivial languages}. These languages are characterized by confluent partially ordered DFAs or by complete, confluent, and self-loop-deterministic partially ordered NFAs (ptNFAs for short). 
  In this paper, we study the complexity of basic questions for several types of partially ordered automata on finite words; namely, the questions of inclusion, equivalence, and ($k$-)piecewise testability. The lower-bound complexity boils down to the complexity of universality. The universality problem asks whether a system recognizes all words over its alphabet. For ptNFAs, the complexity of universality decreases if the alphabet is fixed, but it is open if the alphabet may grow with the number of states. We show that deciding universality for general ptNFAs is as hard as for general NFAs. Our proof is a novel and nontrivial extension of our recent construction for self-loop-deterministic partially ordered NFAs, a model strictly more expressive than ptNFAs. We provide a comprehensive picture of the complexities of the problems of inclusion, equivalence, and ($k$-)piecewise testability for the considered types of automata.
\end{abstract}

\maketitle

\section{Introduction}
  \emph{Partially ordered automata} -- also known as \emph{1-weak}, \emph{very weak}, \emph{linear}, \emph{acyclic} or \emph{extensive automata}~\cite{DaxK08,KufleitnerL11,pin2019} -- are finite automata where the transition relation induces a partial order on states. This restriction on the behaviour of the automata implies that as soon as a state is left during the computation, it is never visited again. In other words, the only cycles of partially ordered automata are self-loops. 
  
  Partially ordered automata attracted attention of many researchers because of their relation to logical, algebraic, and combinatorial charazterizations of languages.
  From the logical perspective, they characterize (Boolean combinations of) fragments of first-order logic on finite words. For an integer $i\ge 1$, the fragment $\Sigma_i$ of the first-order logic FO$[<]$ on finite words consists of all FO$[<]$ formulae in prenex normal form with $i$ blocks of quantifiers, starting with a block of existential quantifiers~\cite{DiekertGK08}.
  From the algebraic perspective, they characterize regular languages whose syntactic monoids possess some algebraic properties. For instance, the syntactic monoid of the language is $J$-trivial or $R$-trivial, where $J$ and $R$ are Green's relations~\cite{green}, see also Pin~\cite{pin2019}.
  From the combinatorial perspective, they characterize some levels of the Straubing-Th\'erien hierarchy~\cite{Straubing81,Therien81}. For an alphabet $\Sigma$, level 0 of the Straubing-Th\'{e}rien hierarchy is defined as $\mathscr{L}(0)=\{\emptyset, \Sigma^*\}$, and for integers $n\geq 0$, the half levels $\mathscr{L}(n+\frac{1}{2})$ consist of all finite unions of languages $L_0 a_1 L_1 a_2 \cdots a_k L_k$, with $k\geq 0$, $L_0,\ldots, L_k\in\mathscr{L}(n)$, and $a_1,\ldots,a_k\in\Sigma$, and the full levels $\mathscr{L}(n+1)$ consist of all finite Boolean combinations of languages from level $\mathscr{L}(n+\frac{1}{2})$. The hierarchy does not collapse on any level~\cite{BrzozowskiK78}; see also Pin~\cite{pin2019} for more details. The Straubing-Th\'erien hierarchy has a close relation to the dot-depth hierarchy~\cite{BrzozowskiK78,CohenB71,Straubing85} and to complexity theory~\cite{Wagner04}.
  We now provide more details and mention related research.
  
  As early as 1972, Simon~\cite{Simon1972} studied a subclass of regular languages called \emph{piecewise testable languages\/} and provided its characterization in terms of partially ordered automata. A regular language over $\Sigma$ is piecewise testable if it is a finite Boolean combination of languages of the form $\Sigma^* a_1 \Sigma^* a_2 \Sigma^* \cdots \Sigma^* a_n \Sigma^*$, where $a_i\in \Sigma$ and $n\ge 0$. If $n\le k$, the language is {\em $k$-piecewise testable}. Simon showed that piecewise testable languages are characterized by \emph{confluent} partially ordered deterministic finite automata (DFAs), where confluence means that for all states $q$ and letters $a$ and $b$, if there are transitions $q\stackrel{a}{\to}q_a$ and $q\stackrel{b}{\to}q_b$, then there is a word $w \in \{a, b\}^*$ that leads the automaton from both $q_a$ and $q_b$ to the same state. 
  Piecewise testable languages are known under many names in the literature. Considering the perpectives described above, piecewise testable languages are known
  as finite boolean combinations of the fragment $\Sigma_1$ of the first-order logic FO$[<]$ on finite words~\cite{DiekertGK08}; 
  as \emph{$J$-trivial languages}, because their syntactic monoids are $J$-trivial~\cite{Simon1972}; or
  as level~1 of the Straubing-Th\'{e}rien hierarchy. 
  This list is, indeed, not exhaustive. Piecewise testable languages appear in many contexts; see, \eg, Karandikar and Schnoebelen~\cite{KarandikarS19} who introduced the \emph{height} of a language and used it to study the expressivity of a two-variable fragment of first-order logic of sequences with the subword ordering. This fragment can only express piecewise testable properties.
  Although we study piecewise testable languages on classical $*$-languages in this paper, it is worth mentioning that piecewise testable languages have been extended from words to trees by Bojanczyk, Segoufin and Straubing~\cite{Bojanczyk:2012}.

  In 1980, Brzozowski and Fich~\cite{BrzozowskiF80} investigated the expressivity of \emph{partially ordered DFAs} and showed that they characterize the class of languages whose syntactic monoid is $R$-trivial. Consequently, the languages of partially ordered DFAs are called \emph{$R$-trivial languages} and are strictly more powerful than piecewise testable languages.
   
  Schwentick et al.~\cite{SchwentickTV01} studied \emph{partially ordered nondeterministic finite automata} (poNFAs) and showed that they characterize the first-order fragment $\Sigma_2$ on finite words or, equivalently, level~$\frac32$ of the Straubing-Th\'{e}rien hierarchy. Partially ordered NFAs are thus strictly more powerful than their deterministic counterpart, partially ordered DFAs.
  Bouajjani et al.~\cite{BMT2001} characterized the class of languages accepted by partially ordered NFAs as languages effectively closed under permutation rewriting, that is, closed under the iterative application of rules of the form $ab\to ba$, and call this class of languages \emph{Alphabetical Pattern Constraints}. 
  
  Schwentick et al.~\cite{SchwentickTV01} further showed that \emph{partially ordered two-way DFAs} coincide with the first-order fragment $\Delta_2$, which consists of all languages where both the language and its complement are $\Sigma_2$-definable. Hence, it is the largest subclass of $\Sigma_2$ closed under complementation. This class of languages is also known as \emph{unambiguous languages}, introduced by Sch\"utzenberger~\cite{Sch76}, who showed that unambiguous languages are exactly those languages whose syntactic monoid is in the variety DA; see also Grosshans et al.~\cite{GrosshansMS17} for their relation to programs over monoids. Lodaya et al.~\cite{LodayaPS08} further characterized the languages of partially ordered two-way DFAs as a fragment of interval temporal logic.
  
  Considering our contribution to the automata characterization of the discussed language classes, we defined \emph{self-loop-deterministic partially ordered NFAs} and showed that they are expressively equivalent to partially ordered DFAs~\cite{mfcs16:mktmmt_full}. A partially ordered NFA is self-loop-deterministic if, in every state, the automaton has never a choice, under the same letter, between staying in the state or leaving the state. Furthermore, we defined and studied a nondeterministic counterpart of confluent partially ordered DFAs recognizing piecewise testable languages called \emph{complete, confluent, and self-loop-deterministic partially ordered NFAs} (or ptNFAs for short)~\cite{ptnfas,dlt15}. 

  There is a constant interest in automata characterizations of the levels of the Straubing-Th\'erien hierarchy (as well as of the fragments of first-order logic), in particular in decidability and complexity of checking membership of a language in a specific level of the hierarchy. Despite a recent progress~\cite{AlmeidaBKK15,Place15,PlaceZ15}, decidability of whether a language belongs to level $k$ of the Straubing-Th\'erien hierarchy is open for $k>\frac{7}{2}$. The recent results were obtained by considering more general problems than membership, and the investigation has brought {\em separability\/} and {\em covering\/} problems into focus. We do not provide more details about these results and rather refer the reader to the literature~\cite{PlaceZ18}. 

  Although we consider only $*$-languages in this paper, we now briefly mention related research for tree- and $\omega$-languages.
  H\'eam~\cite{Heam08} studied tree regular languages and showed that the class of tree regular languages accepted by $\Sigma_2$ formulae is strictly included in the class of languages accepted by partially ordered tree automata. 
  Kufleitner and Lauser~\cite{KufleitnerL11} extended the results of Schwentick et al.~\cite{SchwentickTV01} from finite words to infinite words. They defined \emph{partially ordered two-way B\"{u}chi automata} and characterized their expressivity in terms of first-order logic. Partially ordered B\"{u}chi automata can be used, \eg, to characterize the common fragment of the two temporal logics ACTL and LTL~\cite{Bojanczyk08,Maidl00}.
  They showed that nondeterministic partially ordered two-way B\"{u}chi automata are equivalent to their one-way counterpart, and that their expressivity coincides with the first-order fragment $\Sigma_2$. Considering deterministic partially ordered two-way B\"{u}chi automata, they showed that they characterize the first-order fragment $\Delta_2$, and that deterministic partially ordered two-way B\"{u}chi automata are more expressive than deterministic partially ordered one-way B\"{u}chi automata. Over finite words, the fragment $\Delta_2$ coincides with the fragment FO$^2$ of first-order logic with only two variables~\cite{TherienW98}, and the corresponding languages are exactly unambiguous languages~\cite{PinW97}. The situation is, however, different over infinite words, where the fragment $\Delta_2$ is a proper subclass of FO$^2$, and only \emph{restricted unambiguous languages\/} are $\Delta_2$-definable. 
  They further showed that deterministic partially ordered two-way B\"{u}chi automata are effectively closed under Boolean operations, and discussed the complexity of several problems for partially ordered two-way B\"{u}chi automata, including emptiness, inclusion, universality, and equivalence. The same problems were studied by Lodaya et al.~\cite{LodayaPS10} for partially ordered two-way automata over finite words and by Sistla et al.~\cite{SistlaVW87} for one-way B\"{u}chi automata.

  We finally point out that there is also a kind of partially ordered alternating automata, motivated by applications in specification and verification of nonterminating programs. Namely, Kupferman and Vardi~\cite{KupfermanV01} discussed weak alternating automata of Muller et al.~\cite{MullerSS92}, where the state space is partitioned into partially ordered sets, and the automaton can move from a set only to a smaller set.
  
  Our interest in partially ordered automata has several reasons.
  First, it comes from the supervisory control synthesis of discrete event systems~\cite{KomendaM17,KomendaMS15} and from the verification of properties of discrete-event systems~\cite{MasopustDetectability,MY2017}. Given an automaton modelling a system (manufacturing, technological, etc.) and a regular language describing a specification, the aim of supervisory control is to automatically design a controller such that, running the controller and the system in a closed-loop feedback manner, the controlled system satisfies the prescribed specification. Although the goal of supervisory control synthesis is similar to that of reactive synthesis~\cite{HarelPnueli85}, there are significant differences between the approaches; interested readers are referred to the literature for details~\cite{EhlersLTV17,SchmuckMoorMajumdar18}. For our interest, partially ordered automata are in some sense and on some level of abstraction the simplest models of deadlock-free discrete event systems, and hence convenient to study the lower-bound complexity of the problems under consideration, such as detectability, diagnosability, opacity, etc.~\cite{Bryans2005,CainesGW1988,JacobLF16,Lin2011,OzverenW1990,Ramadge1986,SabooriHadjicostis2007,ShuLin2011,ShuLinYing2007}. 

  Our second motivation comes from database theory and schema languages for XML data, namely from efficient approximate query answering and increasing the user-friendliness of XML Schema. Both problems are motivated by scenarios in which we want to describe something complex by means of a simple language. The technical core of these scenarios consists of {\em separation\/} problems, which are usually of the form ``Given two languages $K$ and $L$, does there exist a language $S$, coming from a family $\mathcal{F}$ of `simple' languages, such that $S$ contains everything from $K$ and nothing from $L$?'' The family $\mathcal{F}$ of simple languages could be, for example, languages definable by a first-order fragment, piecewise testable languages, languages definable by a special class of automata, etc.~\cite{icalp2013,HofmanM15,MartensNNS15,Place18,PlaceRZ13}. The separability technique is further closely related to {\em interpolation\/}, a method providing means to compute separation between good and bad states in program verification~\cite{Craig57,RybalchenkoS10}.

  Our third motivation comes from the evaluation of regular path queries in graph databases. Regular path queries are an important feature of modern query languages, such as SPARQL 1.1, allowing queries about arbitrarily long paths in the graph database. Regular path queries are regular expressions that are matched against labeled directed paths of a graph. The expressions are often of a specific and simple form~\cite{MartensNS09,MartensT18}. Our particular interest is in translating regular path queries to simple (\eg, partially ordered) automata.

  All of our motivations boil down to the complexity questions of basic operations, such as inclusion and equivalence. Since the universality question provides the lower-bound complexity for both inclusion and equivalence, we in particular focus on the complexity of deciding universality.

  Universality is a fundamental question asking whether a given system recognizes all words over its alphabet. The study of universality has a long tradition in formal languages with many applications across computer science, \eg, in knowledge representation and database theory~\cite{BarceloLR:jacm14,CalvaneseGLV03:rpqreasoning,SMKR:elcq14} or in verification~\cite{BaierKatoenBook}. 
  Deciding universality for NFAs is \PSpace-complete~\cite{MeyerS72}, and there are two typical proof techniques for showing hardness. One is based on the reduction from the {\em DFA-union-universality\/} problem~\cite{Kozen77} and the other on the reduction from the {\em word problem\/} for polynomially-space-bounded Turing machines~\cite{AhoHU74}.
  Kozen's~\cite{Kozen77} proof showing \PSpace-hardness of DFA-union universality (actually of its complemented equivalent -- DFA-intersection emptiness) results in DFAs consisting of nontrivial cycles, and these cycles are essential for the proof; indeed, if all cycles of the DFAs were only self-loops, then the problem would be easier (namely, \coNP-complete, see Theorem~\ref{thm4}).

  Deciding universality for partially ordered NFAs has the same worst-case complexity as for general NFAs, even if restricted to binary alphabets~\cite{mfcs16:mktmmt_full}. This could be caused by an unbounded number of nondeterministic steps admitted in partially ordered NFAs -- a partially ordered NFA either stays in the same state or moves to another state. Forbidding this kind of nondeterminism, that is, considering self-loop-deterministic partially ordered NFAs, indeed affects the complexity of universality -- it is \coNP-complete if the alphabet is fixed, but remains \PSpace-complete if the alphabet may grow polynomially with the number of states~\cite{mfcs16:mktmmt_full}. The growth of the alphabet thus, in some sense, compensates for the restricted number of nondeterministic steps. 
  
  In this paper, we study the complexity of deciding universality for ptNFAs -- complete, confluent, and self-loop-deterministic partially ordered NFAs on finite words. Since we use a different definition of these automata in this paper than in our previous work, we first show that these two definitions coincide (Lemma~\ref{lem3.1}). Then we show that deciding universality for ptNFAs is \NL-complete if the input alphabet is unary (Theorem~\ref{thmMainNL}), \coNP-complete if the input alphabet is fixed (Therem~\ref{0ptNFAhard}), and \PSpace-complete in general (Therem~\ref{thmMain}). The proof of Therem~\ref{thmMain} requires a novel and nontrivial extension of our recent construction for self-loop-deterministic partially ordered NFAs~\cite{mfcs16:mktmmt_full}. The proof is based on a construction of its own interest; namely, on a construction of a ptNFA accepting all but a single exponentially long word (Lemma~\ref{exprponfas}). These results are summarized in Table~\ref{table_resultsI}. 
  \begin{table*}\centering
    \ra{1.1}
    \begin{tabular}{@{}llll@{}}\toprule
            & $|\Sigma|=1$
            & $|\Sigma|\ge 2$
            & $\Sigma$ is growing\\
          \midrule
          DFA       & L-c    \cite{Jones75}
                    & \NL-c  \cite{Jones75}
                    & \NL-c  \cite{Jones75}\\
          ptNFA     & \NL-c      (Thm.~\ref{thmMainNL})
                    & \coNP-c    (Thm.~\ref{0ptNFAhard})
                    & \PSpace-c  (Thm.~\ref{thmMain})\\
          rpoNFA    & \NL-c      \cite{mfcs16:mktmmt_full}
                    & \coNP-c    \cite{mfcs16:mktmmt_full}
                    & \PSpace-c  \cite{mfcs16:mktmmt_full}\\
          poNFA     & \NL-c      \cite{mfcs16:mktmmt_full}
                    & \PSpace-c  \cite{mfcs16:mktmmt_full}
                    & \PSpace-c  \cite{AhoHU74} \\
          NFA       & \coNP-c    \cite{StockmeyerM73}
                    & \PSpace-c  \cite{AhoHU74} 
                    & \PSpace-c  \cite{AhoHU74} \\
      \bottomrule
    \end{tabular}
      \caption{Complexity of deciding universality; $\Sigma$ denotes the input alphabet.}
      \label{table_resultsI}
  \end{table*}
  Then we show how to reduce universality to the question of whether the language of a given automaton is $k$-piecewise testable (Lemma~\ref{lemma0k}), and we use this result to obtain the complexity results for deciding $k$-piecewise testability (Theorems~\ref{thmMainUniv} through~\ref{thm30}); see Table~\ref{table1} for an overview of the results. 
  \begin{table*}\centering
    \ra{1.1}
    \begin{tabular}{@{}lllll@{}}\toprule
          & Unary alphabet
          & Fixed alphabet
          & \multicolumn{2}{l}{Arbitrary alphabet}\\
          & $|\Sigma|=1$
          & $|\Sigma|\ge 2$
          & $k\le 3$
          & $k\ge 4$ \\
          \midrule
          DFA   & L-c            (Thm. \ref{DFAlc})
                & \NL-c          (Thm. \ref{kPTfixedDFA})
                & \NL-c          \cite{dlt15}
                & \coNP-c        \cite{KKP}    \\
          ptNFA & \NL-c          (Thm. \ref{thmP})
                & \coNP-c        (Thm. \ref{theorem27})
                & \multicolumn{2}{l}{\PSpace-c (Thm. \ref{thmMainUniv})} \\
        rpoNFA & \NL-c          
                & \coNP-c        \cite{mfcs16:mktmmt_full}
                & \multicolumn{2}{l}{\PSpace-c} \\
          poNFA & \NL-c          (Thm. \ref{poNFAnl})
                & \PSpace-c      (Thm. \ref{thm72})
                & \multicolumn{2}{l}{\PSpace-c} \\
            NFA & \coNP-c        (Thm. \ref{thm30})
                & \PSpace-c      \cite{dlt15}
                & \multicolumn{2}{l}{\PSpace-c \cite{dlt15}} \\
      \bottomrule
    \end{tabular}
    \caption{Complexity of deciding $k$-piecewise testability.}
    \label{table1}
  \end{table*}
  After that, we study the question whether the language of a given automaton is piecewise testable, that is, without the restriction to a specific $k$ (Theorems~\ref{thm74} through~\ref{rpoNFApt}); see Table~\ref{table2} for an overview of the results.
  Finally, we discuss the complexity of the problems of inclusion and equivalence in Section~\ref{IandE}. These results are summarized in Tables~\ref{table3} and~\ref{table4}.

  \begin{table*}\centering
    \ra{1.1}
    \begin{tabular}{@{}llll@{}}\toprule
            & $|\Sigma|=1$
            & $|\Sigma|\ge 2$
            & $\Sigma$ is growing\\
          \midrule
            DFA    & L-c            (Thm. \ref{DFAlcb})
                    & \NL-c          \cite{ChoH91}
                    & \NL-c          \cite{ChoH91} \\
          rpoNFA    & $\checkmark$   (Thm. \ref{thm74})
                    & \coNP-c        (Thm. \ref{rpoNFApt})
                    & \PSpace-c      (Thm. \ref{rpoNFAtoPT})\\
          poNFA     & $\checkmark$   (Thm. \ref{thm74})
                    & \PSpace-c      (Thm. \ref{thm72b})
                    & \PSpace-c       \\
          NFA       & \coNP-c        (Thm. \ref{thm30b})
                    & \PSpace-c      \cite{tm2016}
                    & \PSpace-c      \cite{tm2016} \\
      \bottomrule
    \end{tabular}
    \caption{Complexity of deciding piecewise testability.}
    \label{table2}
  \end{table*}

\begin{table*}\centering
  \ra{1.1}
  \begin{tabular}{@{}lllll@{}}\toprule
                    & \multicolumn{4}{c}{$B$} \\
    \cmidrule{2-5}
    $A$ & DFA & ptNFA \& rpoNFA & poNFA & NFA\\ \midrule
    DFA   & \LogSpace/\NL
          & \NL/\coNP/\PSpace
          & \NL/\PSpace
          & \coNP/\PSpace \\
    ptNFA & \NL
          & \NL/\coNP/\PSpace
          & \NL/\PSpace
          & \coNP/\PSpace \\
    rpoNFA& \NL
          & \NL/\coNP/\PSpace
          & \NL/\PSpace
          & \coNP/\PSpace \\
    poNFA & \NL
          & \NL/\coNP/\PSpace
          & \NL/\PSpace
          & \coNP/\PSpace \\
    NFA   & \NL
          & \NL/\coNP/\PSpace
          & \NL/\PSpace
          & \coNP/\PSpace\\
    \bottomrule
  \end{tabular}
  \caption{Complexity of deciding inclusion $L(A)\subseteq L(B)$ (unary/fixed[/growing] alphabet), all results are complete for the given class.}
  \label{table3}
\end{table*}

\begin{table*}\centering
  \ra{1.1}
  \begin{tabular}{@{}lllll@{}}\toprule
                  & DFA & ptNFA \& rpoNFA & poNFA & NFA \\ \midrule
           DFA    & \LogSpace/\NL
                  & \NL/\coNP/\PSpace
                  & \NL/\PSpace
                  & \coNP/\PSpace \\
        ptNFA     & 
                  & \NL/\coNP/\PSpace
                  & \NL/\PSpace
                  & \coNP/\PSpace \\
        rpoNFA    & 
                  & \NL/\coNP/\PSpace
                  & \NL/\PSpace
                  & \coNP/\PSpace \\
        poNFA     & 
                  & 
                  & \NL/\PSpace
                  & \coNP/\PSpace \\
        NFA       & 
                  & 
                  & 
                  & \coNP/\PSpace \\
    \bottomrule
  \end{tabular}
  \caption{Complexity of deciding equivalence (unary/fixed[/growing] alphabet), the problems are complete for the given classes.}
  \label{table4}
\end{table*}

\section{Preliminaries}
\subsection{Basic Definitions}
  We assume that the reader is familiar with automata and formal language theory~\cite{AhoHU74}. The cardinality of a set $A$ is denoted by $|A|$ and the power set of $A$ by $2^A$. The empty word is denoted by $\eps$. For a word $w=xyz$, $x$ is a {\em prefix}, $y$ a {\em factor\/}, and $z$ a {\em suffix\/} of $w$. Let $\Sigma$ be an alphabet, and let $L_{a_1 a_2 \cdots a_n} = \Sigma^* a_1 \Sigma^* a_2 \Sigma^* \cdots \Sigma^* a_n \Sigma^*$, where $a_i\in\Sigma$ for $i=1,\ldots,n$. A word $v$ is a {\em subword\/} of a word $w$, denoted by $v \preccurlyeq w$, if $w \in L_{v}$. A prefix (factor, suffix, subword) of $w$ is {\em proper\/} if it is different from $w$.

  A {\em nondeterministic finite automaton\/} (NFA) is a quintuple $\A = (Q,\Sigma,\delta,I,F)$, where $Q$ is a finite nonempty set of states, $\Sigma$ is an input alphabet, $I\subseteq Q$ is a set of initial states, $F\subseteq Q$ is a set of accepting states, and $\delta \colon Q\times\Sigma \to 2^Q$ is the transition function that can be extended to the domain $2^Q\times \Sigma^*$ in the usual way. The language {\em accepted\/} by $\A$ is the set $L(\A) = \{w\in\Sigma^* \mid \delta(I,w) \cap F \neq \emptyset\}$. The automaton $\A$ is {\em complete\/} if for every state $q\in Q$ and every letter $a \in \Sigma$, the set $\delta(q,a)$ is nonempty, and it is {\em deterministic\/} (DFA) if $|I|=1$ and $|\delta(q,a)|=1$ for every state $q \in Q$ and every letter $a \in \Sigma$. 

  A {\em path\/} $\pi$ from a state $q_0$ to a state $q_n$ under a word $a_1a_2\cdots a_{n}$, for some $n\ge 0$, is a sequence of states and input symbols $q_0 a_1 q_1 a_2 \ldots q_{n-1} a_{n} q_n$ where $q_{i+1} \in \delta(q_i,a_{i+1})$ for all $i=0,1,\ldots,n-1$. Path $\pi$ is {\em accepting\/} if $q_0\in I$ is an initial state and $q_n\in F$ is an accepting state. We write $q_0 \xrightarrow{a_1a_2\cdots a_{n}} q_{n}$ to denote that there exists a path from $q_0$ to $q_n$ under the word $a_1a_2\cdots a_{n}$. A path is {\em simple\/} if all its states are pairwise distinct. The number of states on the longest simple path of $\A$ that starts in an initial state, decreased by one (\ie, the number of transitions on that path), is the {\em depth\/} of $\A$, denoted by $\depth(\A)$. 
  
\subsection{Partially Ordered Automata}\label{subsec3}
  Let  $\A = (Q,\Sigma,\delta,I,F)$ be an NFA. The reachability relation $\le$ on states is defined by setting $p\le q$ if there is a word $w\in \Sigma^*$ such that $q\in \delta(p,w)$. The NFA $\A$ is {\em partially ordered (poNFA)\/} if the reachability relation $\le$ is a partial order. 

  The NFA $\A$ is {\em confluent\/} provided that for every state $q\in Q$ and every pair of (not necessarily distinct) letters $a, b \in \Sigma$, if $s \in \delta(q,a)$ and $t\in \delta(q,b)$ then there exists a word $w \in \{a, b\}^*$ such that $\delta(s,w) \cap \delta(t,w) \neq \emptyset$; see Figure~\ref{fig_bad_pattern} (left) for an illustration. 
  
  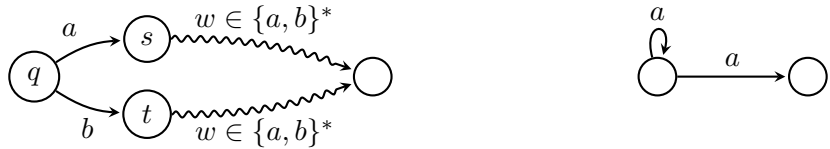
\begin{figure}[b]
    \centering
    \begin{tikzpicture}[baseline,->,auto,>=stealth,shorten >=1pt,node distance=1.5cm,thick,
      state/.style={circle,minimum size=5mm,thick,draw=black,initial text=}]
      \node[state]  (1) {$q$};
      \node         (0) [right of=1]  {};
      \node[state]  (2) [above of=0,node distance=.5cm]  {$s$};
      \node[state]  (3) [below of=0,node distance=.5cm]  {$t$};
      \node[state]  (4) [right of=0,node distance=3cm] {};
      \path
        (1) edge[bend left=15]  node {$a$} (2)
        (1) edge[bend right=15] node[below] {$b$} (3)
        (2) edge[snake arrow,bend left=10] node[above] {$w\in\{a,b\}^*$} (4)
        (3) edge[snake arrow,bend right=10] node[below] {$w\in\{a,b\}^*$} (4)
        ;
    \end{tikzpicture}
    \hspace{3cm}
    \begin{tikzpicture}[baseline,->,auto,>=stealth,shorten >=1pt,node distance=2cm,thick,
      state/.style={circle,minimum size=5mm,thick,draw=black,initial text=}]
      \node[state]  (a) {};
      \node[state]  (aa) [right of=a]  {};
      \path
        (a) edge[loop above] node {$a$} (a)
        (a) edge node {$a$} (aa)
        ;
    \end{tikzpicture}
    \caption{Confluence (left) and the forbidden pattern of rpoNFAs (right).}
    \label{fig_bad_pattern}
  \end{figure}
  
  The poNFA $\A$ is \emph{restricted\/} or {\em self-loop-deterministic\/} ({\em rpoNFA}) if, for every state $q\in Q$ and every letter $a\in \Sigma$, $q\in \delta(q,a)$ implies that $\delta(q,a) = \{q\}$. Restricted poNFAs can thus be defined in terms of forbidden patterns similar to that of Gla{\ss}er and Schmitz~\cite{GlasserS08}; Figure~\ref{fig_bad_pattern} (right) shows the forbidden pattern of rpoNFAs.

  An rpoNFA is a {\em ptNFA\/} if it is complete and confluent. The name ptNFA comes from \emph{piecewise testable}, since ptNFAs characterize piecewise testable languages~\cite{ptnfas,dlt15}. Notice that the disjoint union of two or more ptNFAs is a ptNFA; we use this fact in the proof of Theorem~\ref{thmMain}.
  The violation of the definitional properties of ptNFAs can be tested by several reachability checks, and these properties are \NL-hard even for minimal DFAs~\cite{ChoH91}. Since \NL~=~\coNL, checking whether an NFA is a ptNFA is NL-complete.

\subsection{The Unique Maximal State Property}
  For two states $p$ and $q$, we write $p < q$ if $p\le q$ and $p\ne q$. A state $p$ is {\em maximal\/} if there is no state $q$ such that $p < q$. 
  
  A poNFA $\A$ over $\Sigma$ with the state set $Q$ can be turned into a directed graph $G(\A)$ with the set of vertices $Q$ where a pair $(p,q) \in Q \times Q$ is an edge in $G(\A)$ if there is a transition from $p$ to $q$ in $\A$. For an alphabet $\Gamma \subseteq \Sigma$, we define the directed graph $G(\A,\Gamma)$ with the set of vertices $Q$ by considering only those transitions corresponding to letters in $\Gamma$. 
  Let $\Sigma(p)=\{a\in\Sigma \mid p\xrightarrow{\,a\,} p\}$ denote all letters labeling self-loops in state $p$. 
  We say that $\A$ satisfies the {\em unique maximal state\/} (UMS) property if, for every state $q$ of $\A$, $q$ is the unique maximal state of the connected component of $G(\A,\Sigma(q))$ containing $q$.

  To decide whether, for a given DFA, there exists an equivalent confluent poDFA, or, said differently, whether the given DFA recognizes a piecewise testable language, Kl\'ima and Pol\'ak~\cite{KlimaP13} check the confluence property of the minimal DFA while Trahtman~\cite{Trahtman2001} checks the UMS property.
  Both properties have their advantages and an effect on algorithmic complexity. While Trahtman's algorithm runs in quadratic time with respect to the number of states and in linear time with respect to the size of the alphabet, Kl\'ima and Pol\'ak's algorithm swaps the complexities; it runs in linear time with respect to the number of states and in quadratic time with respect to the size of the alphabet. From the computational complexity view, the problem is \NL-complete~\cite{ChoH91}. 
  
  Although the UMS property and the confluence property coincide on DFAs, they differ on NFAs. Consider the automaton $\B$ in Figure~\ref{fig1}. One can verify that $\B$ is confluent. However, $\B$ does not satisfy the UMS property. Indeed, for state $0$, the connected component is $G(\B,\Sigma(0))=G(\B,\{a,b\})=G(\B)$, \ie, the component is the whole automaton. Therefore, state $0$ is not a maximal state of this component; in particular, it is not the unique maximal state. 
  Furthermore, notice that $\B$ accepts a language that is not piecewise testable. Indeed, there is an infinite sequence of words $a,ab,aba,abab,\ldots$ where accepted words alternate with non-accepted words.
  The corresponding minimal DFA therefore must have a nontrivial cycle that contains at least one accepting and one non-accepting state. This, however, means that the minimal DFA is not partially ordered, and hence it cannot recognize a piecewise testable language. In summary, we find that the confluence property on its own is not suitable for characterizing piecewise testable languages for poNFAs. Intuitively, the problem comes from the presence of the forbidden pattern of rpoNFAs (see Figure~\ref{fig_bad_pattern}), which may fool confluence but not the UMS property.  
  
  \begin{figure}
    \centering
    \begin{tikzpicture}[baseline,->,auto,>=stealth,shorten >=1pt,node distance=2.5cm,thick,
      state/.style={circle,minimum size=7mm,thick,draw=black,initial text=},
      every node/.style={fill=white,font=\small}]
      \node[state,initial]    (0) {$0$};
      \node[state,accepting]  (1) [right of=0] {$1$};
      \node[state]            (2) [right of=1] {$2$};
      \path
        (0) edge node {$a$} (1)
        (1) edge node {$b$} (2)
        (0) edge[loop above] node {$a,b$} (0)
        (1) edge[loop above] node {$a$} (1)
        (2) edge[loop above] node {$a,b$} (2)
        ;
    \end{tikzpicture}
    \caption{A confluent automaton $\B$ accepting a non-piecewise testable language.}
    \label{fig1}
  \end{figure}
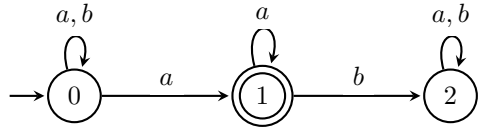
  
  Recall that ptNFAs are complete and confluent rpoNFAs. The next lemma gives an alternative characterization of ptNFAs as poNFAs that are complete and satisfy the UMS property, which was used as the primary definition of ptNFAs in our previous work~\cite{ptnfas}.
  
  \begin{lem}\label{lem3.1}
    Partially ordered NFAs that are complete and satisfy the UMS property are exactly ptNFAs.
  \end{lem}
  \begin{proof}
    First, we show that if $\A$ is partially ordered, complete, and satisfies the UMS property, then $\A$ is a ptNFA, \ie, that $\A$ is a complete and confluent rpoNFA. Indeed, the presence of the forbidden pattern of rpoNFAs violates the UMS property, \cf\ Figure~\ref{fig1}. Therefore, $\A$ is a complete rpoNFA. It remains to show that $\A$ is confluent. To this aim, let $r$ be a state of $\A$, and let $a$ and $b$ be letters of the alphabet of $\A$ ($a=b$ not excluded) such that $r \xrightarrow{a} s$, $r\xrightarrow{b} t$, and $s \neq t$. Let $s'$ and $t'$ be any maximal states reachable from $s$ and $t$ in the component $G(\A,\{a,b\})$, respectively. Since $\{a,b\}\subseteq \Sigma(s')$ by the definition of $s'$ and the assumption that $\A$ is complete, $t'$ is also in the component $G(\A,\Sigma(s'))$ containing $s'$. However, by the UMS property of $\A$, $s'$ is the unique maximal state of that component, and therefore there must be a path from $t'$ to $s'$ under $\Sigma(s')$. Similarly, interchanging $s'$ and $t'$, there must be a path from $s'$ to $t'$ under $\Sigma(t')$. Since $\A$ is partially ordered, the paths give that $t'\le s'$ and $s'\le t'$, \ie, $s'=t'$. To show that $\A$ is confluent, let $u\in\{a,b\}^*$ be a word labeling the path from $s$ to $s'$ in $G(\A,\{a,b\})$. Since $\A$ is complete, there is a path from state $t$ under $u$ to a state $p$. Let $v\in\{a,b\}^*$ be a word labeling a path from $p$ to a maximal state, $p'$, in $G(\A,\{a,b\})$; such a path exists, since $\A$ is complete and partially ordered. Let $w=uv$, and notice that $v$ can be read in $s'$. Then we have that $s\xrightarrow{w} s'$, $t\xrightarrow{w} p'$, and, from the above, since $s'$ and $t'$ were arbitrary, we have that $s'=p'$. Thus, $\A$ is confluent. 

    To prove the other direction, assume that $\A$ is a ptNFA, \ie, that $\A$ is a complete and confluent rpoNFA. Obviously, $\A$ is a complete poNFA, and it remains to show that $\A$ satisfies the UMS property. For the sake of contradiction, suppose that $\A$ does not satisfy the UMS property, \ie, that there is a state $q$ in $\A$ such that the component $G(\A,\Sigma(q))$ containing $q$ has at least two different maximal states. Let $r$ be a biggest state with respect to the partial order on states of that component such that at least two different maximal states, say $s\neq t$, are reachable from $r$ under the alphabet $\Sigma(q)$; such a state exists by assumption. Obviously, $r\notin \{s,t\}$. Let $s' \in \delta(r,a)$ and $t' \in \delta(r,b)$ be two different states on the paths from $r$ to $s$ and from $r$ to $t$, respectively, for some letters $a, b \in \Sigma(q) \setminus \Sigma(r)$; notice that such letters $a$ and $b$ exist because $\A$ is an rpoNFA, and that $a=b$ is not excluded. Then $r < s'$ and $r<t'$. Since $\A$ is confluent, there exists $r'$ such that $r'\in \delta(s',w) \cap \delta(t',w)$, for some word $w \in \{a,b\}^*$. Let $r''$ denote a maximal state that is reachable from $r'$ in $G(\A,\Sigma(q))$. Then there are three cases: 
    (i) if $r'' = s$, then $r < t'$ and both $s$ and $t$ are reachable from $t'$ under $\Sigma(q)$, which is a contradiction with the choice of $r$ as the biggest state of the component with this property;
    (ii) $r'' = t$ yields a contradiction with the choice of $r$ as in (i) by interchanging $t'$ and $s'$;
    (iii) similarly, $r'' \notin \{s,t\}$ yields a contradiction with the choice of $r$, since $r < s'$ and $r < t'$ and, \eg, $s$ and $r''$ are two different maximal states of the component $G(\A,\Sigma(q))$ containing $q$ that are both reachable from $s' > r$. 
    Thus, $\A$ satisfies the UMS property, which completes the proof.
  \end{proof}

\section{Complexity of Deciding Universality for ptNFAs}
  We now study the complexity of deciding universality for ptNFAs. As mentioned before, these automata are known to characterize the piecewise testable languages. Our results of this section in the context of existing results are summarized in Table~\ref{table_results}.
  \begin{table*}\centering
    \ra{1.1}
    \begin{tabular}{@{}llllllll@{}}\toprule
            & ST
            & $|\Sigma|=1$
            & $|\Sigma|\ge 2$
            & $\Sigma$ is growing\\
          \midrule
          DFA       &
                    & L-c    \cite{Jones75}
                    & \NL-c  \cite{Jones75}
                    & \NL-c  \cite{Jones75}\\
          ptNFA     & $1$
                    & \NL-c      (Thm.~\ref{thmMainNL})
                    & \coNP-c    (Thm.~\ref{0ptNFAhard})
                    & \PSpace-c  (Thm.~\ref{thmMain})\\
          rpoNFA    &
                    & \NL-c      \cite{mfcs16:mktmmt_full}
                    & \coNP-c    \cite{mfcs16:mktmmt_full}
                    & \PSpace-c  \cite{mfcs16:mktmmt_full}\\
          poNFA     & $\frac{3}{2}$
                    & \NL-c      \cite{mfcs16:mktmmt_full}
                    & \PSpace-c  \cite{mfcs16:mktmmt_full}
                    & \PSpace-c  \cite{AhoHU74} \\
          NFA       &
                    & \coNP-c    \cite{StockmeyerM73}
                    & \PSpace-c  \cite{AhoHU74} 
                    & \PSpace-c  \cite{AhoHU74} \\
      \bottomrule
    \end{tabular}
      \caption{Complexity of deciding universality; ST stands for the corresponding level of the Straubing-Th\'erien hierarchy; $\Sigma$ denotes the input alphabet.}
      \label{table_results}
  \end{table*}

  For unary alphabets, deciding universality for ptNFAs is solvable in polynomial time~\cite{mfcs16:mktmmt_full}. We now improve the result and show that the problem is \NL-complete. 
  
  \begin{thm}\label{thmMainNL}
    Deciding universality for ptNFAs over a unary alphabet is \NL-complete.
  \end{thm}
  \begin{proof}
    The problem is in \NL even for unary poNFAs~\cite{mfcs16:mktmmt_full}. 
    To prove hardness, we reduce the DAG-reachability problem~\cite{Jones75}. Let $G$ be a directed acyclic graph with $n$ nodes, and let $s$ and $t$ be two nodes of $G$ such that there is no edge from $t$. We define a ptNFA $\A$ as follows. With each node of $G$, we associate a state in $\A$. Whenever there is an edge from $i$ to $j$ in $G$, we add a transition $i\xrightarrow{a} j$ to $\A$. The initial state of $\A$ is $s$ and all states are accepting. The automaton is obviously an rpoNFA, because there  are no (self-)loops. To make it complete and confluent, we add $n-1$ new non-accepting states $f_1,\ldots,f_{n-1}$ together with transitions $f_i \xrightarrow{a} f_{i+1}$, for $i=1,\ldots,n-2$, $f_{n-1} \xrightarrow{a} t$, $t\xrightarrow{a} t$, and, for every state $q\notin\{t,f_1,\ldots,f_{n-1}\}$, we add the transition $q\xrightarrow{a} f_1$. The resulting automaton is a ptNFA.
    
    We now show that $\A$ is universal if and only if $t$ is reachable from $s$ in $G$. If $t$ is reachable from $s$ in $G$, then $L(\A)=\{a\}^*$, since $t$ is reachable from $s$ via states corresponding to nodes of $G$, which are all accepting in $\A$. If $t$ is not reachable from $s$ in $G$, then $t$ is reachable from $s$ in $\A$ via a path $s \xrightarrow{a^k} q \xrightarrow{a} f_1\xrightarrow{a} f_2\xrightarrow{a} \ldots \xrightarrow{a} f_{n-1} \xrightarrow{a} t$ for any state $q$ corresponding to a node of $G$ different from $t$ reachable from $s$ in $G$. We show that $a^{n-1}$ does not belong to $L(\A)$. The shortest path from state $s$ to state $t$ in $\A$ is of length $n$ for $q=s$. Thus, any word accepted in $t$ is of length at least $n$. On the other hand, every word accepted in a state corresponding to a node of $G$ different from $t$ is of length at most $n-2$, since there are $n-1$ such states and $\A$ is acyclic on those states. This gives that $a^{n-1}$ is not accepted by $\A$, and hence $L(\A)$ is not universal.
  \end{proof}

  We next show that if the alphabet is fixed, deciding universality for ptNFAs is \coNP-complete, and that hardness holds even if restricted to binary alphabets. Our proof is based on a construction that shows non-equivalence of regular expressions with union and concatenation to be \NP-complete, even if one of the expressions has the form $\Sigma^n$ for some fixed $n$~\cite{Hunt73,StockmeyerM73}.
  
  \begin{thm}\label{0ptNFAhard}
    Deciding universality for ptNFAs over a fixed alphabet is \coNP-complete even if the alphabet is binary.
  \end{thm}
  \begin{proof}
    Membership follows from the membership for rpoNFAs~\cite[Corollary~24]{mfcs16:mktmmt_full}.
    To show hardness, we reduce DNF validity\footnote{\label{ft1}A \emph{(boolean) formula} is built from propositional variables; operators conjunction, disjunction, and negation; and parentheses. A formula is \emph{valid} if it is true for every assignment of $1$ ({\it true}) and $0$ ({\it false}) to its variables. A \emph{literal} is a variable or its negation. A formula is in \emph{disjunctive normal form} (DNF) if it is a disjunction of one or more conjunctions of literals; \eg, $\varphi = (x\land y \land z) \lor (\neg x\land y \land z)$ is a formula in DNF consisting of two conjunctions $x\land y \land z$ and $\neg x\land y \land z$. Given a formula in DNF, the DNF validity problem asks whether the formula is valid. The formula $\varphi$ is not valid; it is not true for, \eg, $(x,y,z)=(0,1,0)$.}. 
    Let $U=\{x_1,\ldots,x_n\}$ be a set of variables and $\varphi = \varphi_1 \lor \cdots \lor \varphi_m$ be a formula in DNF, where every $\varphi_i$ is a conjunction of literals. Without loss of generality, we may assume that no $\varphi_i$ contains both $x$ and $\neg x$. To illustrate our construction, we use the formula $(x\land y) \lor (\neg x \land z)$ over three variables $x,y,z$. For every $i=1,\ldots,m$, we define a regular expression $\beta_i = \beta_{i,1}\ldots\beta_{i,n}$, where 
    \[
      \beta_{i,j} = \left\{
        \begin{array}{ll}
          0+1 & \text{ if neither } x_j \text{ nor } \neg x_j \text{ appear in } \varphi_i\\
          0   & \text{ if } \neg x_j \text{ appears in } \varphi_i\\
          1   & \text{ if } x_j \text{ appears in } \varphi_i
        \end{array}
        \right.
    \]
    for $j=1,\ldots,n$. For our formula $(x\land y) \lor (\neg x \land z)$, we obtain $\beta_1=11(0+1)$ and $\beta_2=0(0+1)1$. We now define a regular expression $\beta = \sum_{i=1}^{m} \beta_{i}$ as the alternative of all expressions $\beta_{i}$. Then $w\in L(\beta)$ if and only if $w$ corresponds to a truth assignment that satisfies some $\varphi_i$. That is, $L(\beta) = \{0,1\}^n$ if and only if $\varphi$ is valid. By construction, the length of each $\beta_{i}$ is proportional to $n$.

    We now construct a ptNFA $\M$ as follows (the transitions are the minimal sets satisfying the definitions). The initial state of $\M$ is state $0$. For every $\beta_{i}$, we construct a deterministic path consisting of $n+1$ states $\{q_{i,0},q_{i,1},\ldots,q_{i,n}\}$ and transitions $q_{i,\ell} \xrightarrow{\beta_{i,\ell+1}} q_{i,\ell+1}$, for $0\le \ell < n$, $q_{i,0} = 0$, and $q_{i,n}$ is accepting, where, for $\beta_{i,\ell+1}=0+1$, $q_{i,\ell} \xrightarrow{\beta_{i,\ell+1}} q_{i,\ell+1}$ denotes two transitions: $q_{i,\ell} \xrightarrow{~0~} q_{i,\ell+1}$ and $q_{i,\ell}\xrightarrow{~1~} q_{i,\ell+1}$. This ensures that $\M$ accepts $L(\beta)$. The overall construction for our formula $(x\land y) \lor (\neg x \land z)$ is illustrated in Figure~\ref{exampleProof32}.
    \begin{figure}
      \centering
        \includegraphics[angle=0,scale=.7]{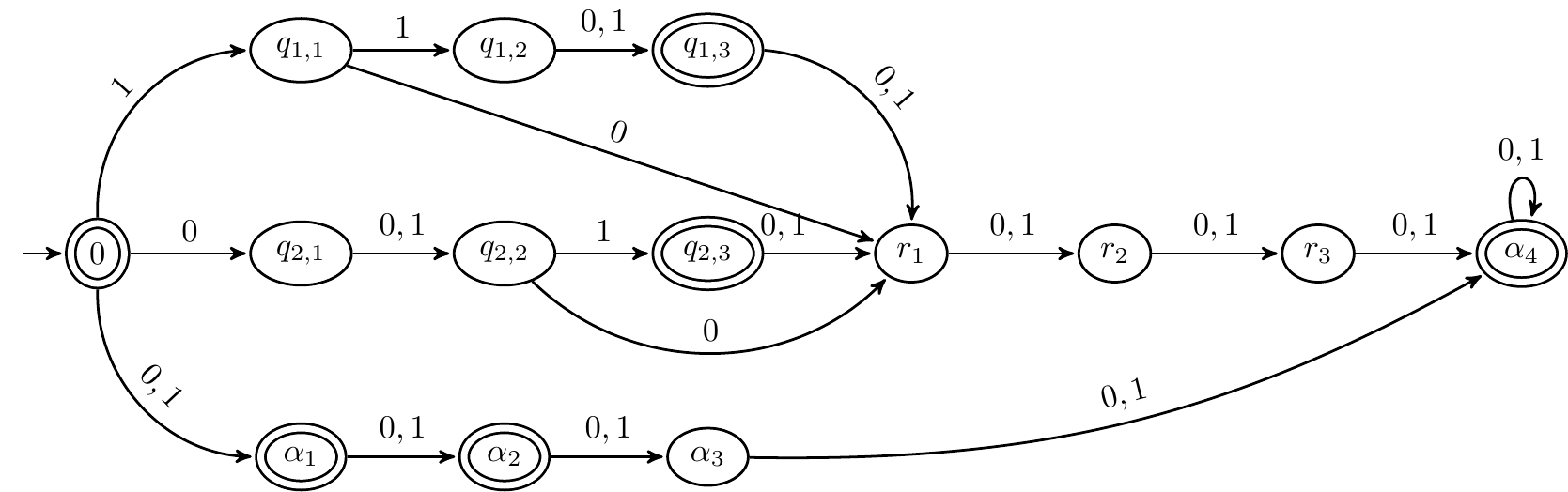}
        \caption{The ptNFA for the formula $(x\land y) \lor (\neg x \land z)$ from the proof of Theorem~\ref{0ptNFAhard}.}
        \label{exampleProof32}
    \end{figure} 

    To accept all words of length different from $n$, we add $n+1$ states $\{\alpha_1,\alpha_2,\ldots,\alpha_{n+1}\}$ and transitions $\alpha_\ell \xrightarrow{~a~} \alpha_{\ell+1}$, for $0\le \ell < n+1$, and $\alpha_{n+1}\xrightarrow{~a~} \alpha_{n+1}$, where $a\in \{0,1\}$, $\alpha_0 = 0$, and $\{\alpha_0,\alpha_1,\ldots\alpha_{n+1}\}\setminus\{\alpha_n\}$ are accepting.
    
    Finally, to make the automaton complete and confluent, we add $n$ non-accepting states $\{r_1,\ldots,r_n\}$ and transitions $r_i\xrightarrow{~a~} r_{i+1}$, for $1\le i<n$, and $r_n \xrightarrow{~a~} \alpha_{n+1}$, where $a\in \{0,1\}$. These states are used to complete $\M$ by adding a transition from every state $q$ to $r_1$ under $a$ if $a$ is not defined in $q$. Then completeness and the fact that state $\alpha_{n+1}$ is reachable from every state implies confluence of the automaton. Notice that the states $r_1,\ldots,r_n$ and the corresponding transitions do not accept any word of length $n$ that does not belong to $L(\beta)$. 
    
    Altogether, the accepting states of $\M$ are states $\{0\}\cup \{q_{1,n},\ldots,q_{m,n}\}\cup \{\alpha_1,\ldots\alpha_{n+1}\}\setminus\{\alpha_n\}$. Thus, $\M$ is a ptNFA accepting the language $L(\M) = L(\beta) \cup \{w \in \{0,1\}^* \mid |w| \neq n\}$, and hence $L(\M)=\{0,1\}^*$ if and only if $L(\beta) = \{0,1\}^n$, which is if and only if $\varphi$ is valid.
  \end{proof}
  
  If the alphabet may grow polynomially with the number of states, there are basically two approaches how to tackle the universality problem for ptNFAs to show \PSpace-hardness: (1) to use a reduction from Kozen's DFA-union-universality problem~\cite{Kozen77}, or (2) to use a reduction from the word problem of polynomially-space-bounded Turing machines \`a la Aho, Hopcroft and Ullman~\cite{AhoHU74}. We discuss these two options in the following two subsections. In the first subsection, we show that the partially ordered variant of the DFA-union-universality problem is easier than its general counterpart, and hence not suitable to prove \PSpace-hardness of deciding universality for ptNFAs. Therefore, in the second subsection, we use the reduction from the word problem of a polynomially-space-bounded Turing machine to prove the result.
  
\subsection{Partially Ordered DFA Union Universality}
  To use the union-universality problem for our purposes, we would need to use partially ordered DFAs (poDFAs) rather than general DFAs to ensure that the union of the DFAs is partially ordered. We now show that the difficulty of the DFA-union-universality problem comes from nontrivial cycles, and hence its partially-ordered variant is easier unless \PSpace~=~\NP.

  We consider the complemented equivalent of the problem -- the {\em DFA-intersection emptiness}. It asks, given $n$ DFAs, whether the intersection of their languages is empty. Clearly, the union of $n$ DFA languages is universal if and only if the intersection of their complements is empty.
  \begin{thm}\label{thm4}
    The intersection-emptiness problem for poDFAs/poNFAs is \coNP-complete. It is \coNP-hard even if the alphabet is binary.\footnote{We thank G. Zetzsche and O. Kl\'ima, who suggested the membership proof.}
  \end{thm}
  \begin{proof}
    We show membership in \coNP for poNFAs and \coNP-hardness for poDFAs.

    Let $\A_1,\ldots,\A_n$ be poNFAs and assume that $w\in \bigcap_{i=1}^{n} L(\A_i)$. Let $k_i$ be the depth of $\A_i$ and consider a fixed path of $\A_i$ accepting $w$. Along this path, we mark (at most) $k_i$ letters of $w$ that cause the change of state of $\A_i$. Doing this for all of the $n$ automata, we mark at most $k_1+k_2+\cdots+k_n$ letters in $w$. Since all non-marked letters of $w$ correspond to self-loops in the automata, we can remove them to obtain a subword $w'$ of $w$ accepted by every $\A_i$ that is of length at most $\sum_{i=1}^{n} k_i$, which is polynomial in the size of the input. Thus, if the intersection is nonempty, there is a polynomial certificate. Therefore, the intersection-emptiness problem is in \coNP.
    
    To show hardness, we reduce DNF validity. Let $\varphi$ be a formula in DNF with $n$ variables and $m$ conjunctions of literals. For $i=1,\ldots,m$, we define an expression $\beta_i$ as in the proof of Theorem~\ref{0ptNFAhard}. Since every $\beta_i$ represents a finite language, we construct a poDFA recognizing $L(\beta_i)$ and take its complement, denoted by $\A_i$, which is also a poDFA. Then $\{0,1\}^n \cap \bigcap_{i=1}^{m} L(\A_i) = \emptyset$ if and only if $\bigcup_{i=1}^{m} L(\beta_{i}) = \{0,1\}^n$, which is if and only if $\varphi$ is valid.
  \end{proof}

  Thus, we have the following corollary.
  \begin{cor}
    The poDFA-union-universality problem is \coNP-complete.
    \qed
  \end{cor}
  
  We point out that the previous result cannot be further strengthened to the case of unary alphabets, since the intersection-emptiness problem for unary poNFAs can be solved in polynomial time. Indeed, if there is a word in the intersection $\bigcap_{i=1}^{n} L(\A_i)$, then it is a prefix of the word $a^{k_1+\cdots+k_n}$, where the $k_i$'s are as in the proof of Theorem~\ref{thm4}, and this word is of polynomial length.

\subsection{Complexity of Deciding Universality for ptNFAs}
  In this section, we show that the universality problem for ptNFAs is \PSpace-complete if the alphabet may grow polynomially with the number of states of the automaton. The proof is a novel and nontrivial extension of our recent proof showing a similar result for self-loop-deterministic poNFAs~\cite{mfcs16:mktmmt_full}. 
  
  To prove our result, we use the idea of Aho, Hopcroft and Ullman~\cite{AhoHU74} to take, for a polynomial $p$, a $p$-space-bounded deterministic Turing machine $\M$ together with an input $x$, and to encode the computations of $\M$ on $x$ as words over some alphabet $\Sigma$, where $\Sigma$ depends on $\M$. One then constructs a regular expression (or an NFA) $R_x$ representing all computations that do not encode an accepting run of $\M$ on $x$. That is, $L(R_x)=\Sigma^*$ if and only if $\M$ does not accept $x$~\cite{AhoHU74}.

  The form of $R_x$ is relatively simple, consisting of a union of expressions of the form
  \begin{equation}\label{eq_part}
    \Sigma^* \, K \, \Sigma^*
  \end{equation}
  where $K$ is a finite language of words of length $O(p(|x|))$. 
  Intuitively, $K$ encodes possible ``violations'' of a correct computation of $\M$ on $x$, such as the initial configuration does not contain the input $x$, or the step from a configuration to the next one does not correspond to a rule of $\M$. These checks are local, involving at most two consecutive configurations of $\M$, each of polynomial size. Hence they can be encoded as the finite language $K$. 
  The initial segment $\Sigma^*$ of \eqref{eq_part} nondeterministically guesses a position of the computation where a violation encoded by $K$ occurs, and the last $\Sigma^*$ reads the rest of the word if the violation check was successful.
  However, this idea cannot be directly used to prove our result for two reasons:
  \begin{enumerate}
    \item\label{probi} Although expression~\eqref{eq_part} can easily be translated to a poNFA, it is not true for ptNFAs because the translation of the leading part $\Sigma^* K$ may not be self-loop-deterministic; 
    \item\label{probii} The constructed poNFA may be incomplete and its ``standard'' completion by adding the missing transitions to a new sink state may violate confluence.
  \end{enumerate}
 
  We addressed problem \eqref{probi} in our previous work~\cite{mfcs16:mktmmt_full} where we proved \PSpace-hardness of deciding universality for rpoNFAs. However, the constructed rpoNFA is not a ptNFA, and because of different expressive powers, it is not always possible to complete an rpoNFA to obtain a ptNFA. To solve problem \eqref{probii}, we use an observation that the length of the encoding of a computation of $\M$ on $x$ is at most exponential with respect to the size of $\M$ and $x$, and hence we could use an exponentially long word to make the automaton confluent in a similar way as we did in our previous constructions above. Since such a word cannot be constructed by a polynomial-time reduction, we need to encode it by a ptNFA, which exists and is of polynomial size as we show in Lemma~\ref{exprponfas} -- there we construct, in polynomial time, a ptNFA $\A_{n,n}$ that accepts all words but a single one, $W_{n,n}$, of exponential length. The automaton consists of two parts (the upper and lower parts in Figure~\ref{ptnfa3}). The upper part (together with state {\em max}) is the rpoNFA we used to tackle problem~\eqref{probi} in our previous work~\cite{mfcs16:mktmmt_full}. The lower part not only makes the rpoNFA complete and confluent, but it also encodes an exponentially long word that we use to tackle problem~\eqref{probii}.
  
  In our proof we do not get the same language as defined by the regular expression $R_x$, but the language of the constructed ptNFA is universal if and only if the language of $R_x$ is, which suffices for our reduction.

  It is not hard to see that an automaton that is complete and has a single maximal state reachable from every state must also be confluent. We use this fact to simplify the proof.
  
  Thus, the first step of the construction is to construct the ptNFA $\A_{n,n}$, which accepts all words but the single word $W_{n,n}$ of exponential length. This automaton is the core of the proof. The considered language is the same as in our previous work~\cite[Lemma~17]{mfcs16:mktmmt_full}, where the constructed automaton is an rpoNFA that is not a ptNFA. As already pointed out above, that rpoNFA is formed by the upper part of Figure~\ref{ptnfa3} together with state {\em max}, \cf~Corollary~\ref{cor_nonacc}. In the following lemma, we present the basic idea how to transform this rpoNFA into a ptNFA. This idea is used and further developed in the proof of Theorem~\ref{thmMain}. On an intuitive level, the states of the lower part of Figure~\ref{ptnfa3}, except for state {\em max}, are used to make the rpoNFA (as well as the rpoNFAs constructed in the proof of Theorem~\ref{thmMain}) complete and confluent, and hence a ptNFA. Recall that because of a weaker expressivity of ptNFAs, this transformation is not possible in general.
  
  \begin{lem}\label{exprponfas}
    For all integers $k,n\geq 1$, there exists a ptNFA $\A_{k,n}$ over an $n$-letter alphabet with $n(2k+1)+1$ states, such that the unique non-accepted word of $\A_{k,n}$ is of length $\binom{k+n}{k}-1$.
  \end{lem}
  \begin{proof}
    For positive integers $k$ and $n$, we recursively define words $W_{k,n}$ over the alphabet $\Sigma_n = \{a_1,a_2,\ldots, a_n\}$ as follows. For the base cases, we set $W_{k,1} = a_1^k$ and $W_{1,n} = a_1a_2\ldots a_n$. The cases for $k,n >1$ are defined recursively by setting
    \begin{align*}
      W_{k,n} 
        & = W_{k,n-1}\, a_{n}\, W_{k-1,n} \\
        & = W_{k,n-1}\, a_n\, W_{k-1,n-1}\, a_n\, W_{k-2,n}\\
        &~~\, \vdots\\
        & = W_{k,n-1}\, a_n\, W_{k-1,n-1}\, a_n\, \cdots\, a_n\, W_{1,n-1}\, a_n\,.
    \end{align*}
    The length of $W_{k,n}$ is $\binom{k+n}{n}-1$~\cite{dlt15}. Notice that the letter $a_n$ appears exactly $k$ times in $W_{k,n}$. We further set $W_{k,n}=\eps$ whenever $kn=0$, since this is useful for defining $\A_{k,n}$ below. 

    We construct a ptNFA $\A_{k,n}$ over $\Sigma_n$ that accepts the language $\Sigma_n^* \setminus \{W_{k,n}\}$. For $n=1$ and $k\ge 0$, let $\A_{k,1}$ be a DFA for $\{a_1\}^* \setminus \{a_1^k\}$ with $k$ additional unreachable states used to address confluence of problem~\eqref{probii} and included here for uniformity (see Corollary~\ref{WnnStructure}). Thus, $\A_{k,1}$ consists of $2k+1$ states of the form $(i;1)$ and a state {\em max} together with the given $a_1$-transitions, see Figure~\ref{ptnfa3} for an illustration. All states but $(i;1)$, for $i=k,\ldots,2k$, are accepting, and $(0;1)$ is initial. All undefined transitions in Figure~\ref{ptnfa3} go to state {\em max}.

    Given a ptNFA $\A_{k,n-1}$, we recursively construct $\A_{k,n}$ as defined next. The construction for $n=3$ is illustrated in Figure~\ref{ptnfa3}. We obtain $\A_{k,n}$ from $\A_{k,n-1}$ by adding $2k+1$ states $(0;n),(1;n),\ldots,(2k;n)$, where $(0;n)$ is added to the initial states, and all states $(i;n)$ with $i< k$ are added to the accepting states. The automaton $\A_{k,n}$ therefore has $n(2k+1) + 1$ states.
    \begin{figure} 
      \centering
        \includegraphics[angle=0,scale=.7]{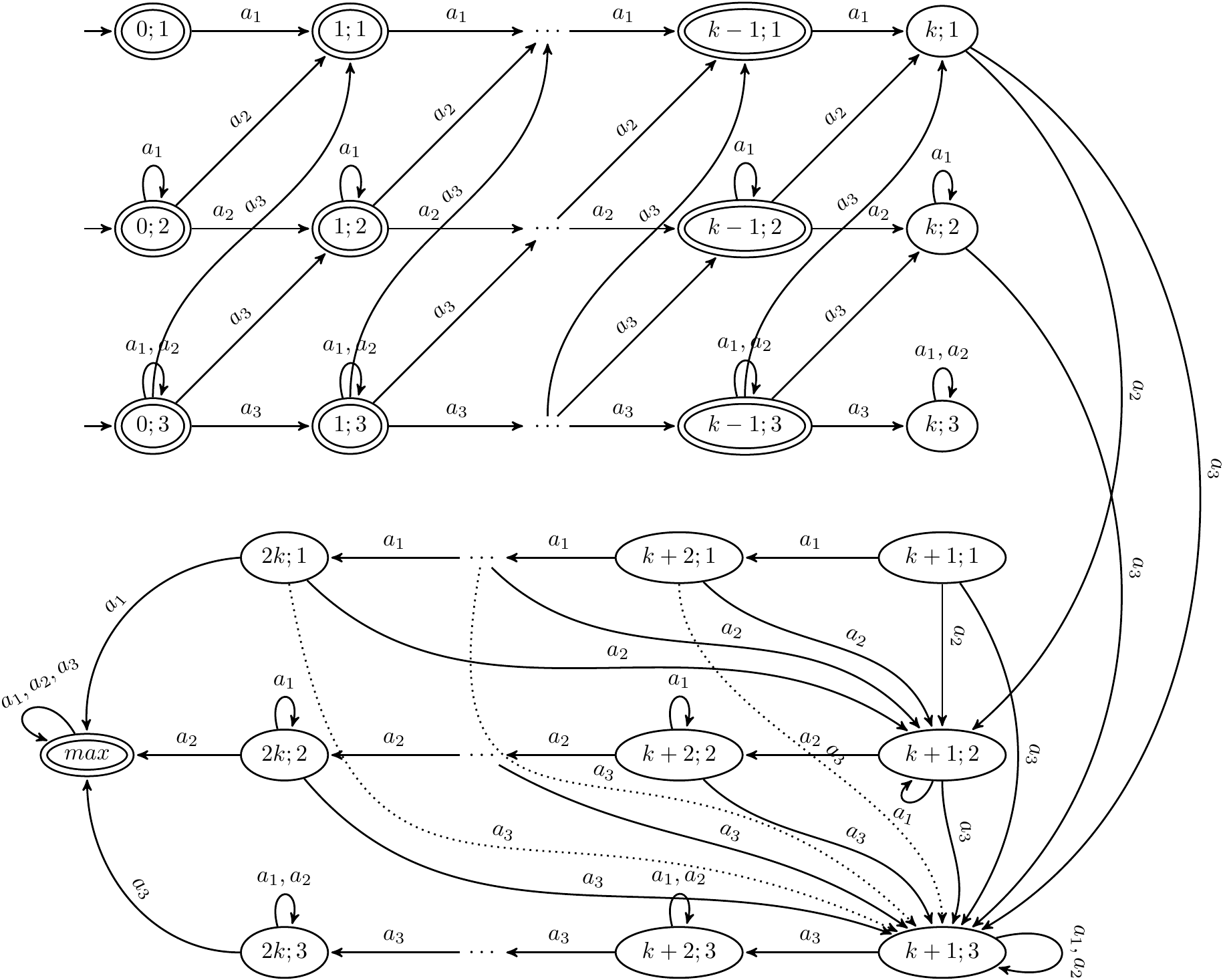}
        \caption{The ptNFA $\A_{k,3}$ with $3(2k+1) + 1$ states; all undefined transitions go to state {\em max}; for readability, transitions $(k+i;1)\stackrel{a_3}{\to}(k+1;3)$ for $i=2,3,\ldots,k$, are dotted.}
        \label{ptnfa3}
    \end{figure}
    The additional transitions of $\A_{k,n}$ consist of the following groups:
    \begin{enumerate}
      \item\label{r1} Self-loops $(i;n)\xrightarrow{a_j}(i;n)$ for $i\in\{0,1,\ldots,2k\}$ and $a_j=a_1,a_2,\ldots,a_{n-1}$;
        
      \item\label{r2} Transitions $(i;n)\xrightarrow{a_n}(i+1;n)$ for $i\in\{0,1,\ldots,2k-1\}\setminus\{k\}$;
        
      \item\label{r3} Transitions $(k,n)\xrightarrow{a_n} max$, $(2k,n)\xrightarrow{a_n} max$, and the self-loop $max \xrightarrow{a_n} max$;
        
      \item\label{r4} Transitions $(i;n)\xrightarrow{a_n}(i+1;m)$ for $i=0,1,\ldots,k-1$ and $m=1,\ldots,n-1$;
        
      \item\label{r5} Transitions $(i;m)\xrightarrow{a_n} max$ for every accepting state $(i;m)$ of $\A_{k,n-1}$;
        
      \item\label{r6} Transitions $(i;m)\xrightarrow{a_n} (k+1,n)$ for every non-accepting state $(i;m)$ of $\A_{k,n-1}$.
    \end{enumerate}

    By construction, $\A_{k,n}$ is complete and partially ordered. It satisfies the UMS property because if there is a self-loop in a state $q\neq max$ under a letter $a$, then there is no other incoming or outgoing transition of $q$ under $a$. This means that the component of the graph $G(\A_{k,n},\Sigma(q))$ containing $q$ is only state $q$, which is indeed the unique maximal state. Hence, it is a ptNFA. 
      
    We show that $\A_{k,n}$ accepts $\Sigma_n^*\setminus \{W_{k,n}\}$. The transitions {\eqref{r1}}, {\eqref{r2}}, and {\eqref{r3}} ensure acceptance of every word that does not contain exactly $k$ occurrences of $a_n$.
    The transitions {\eqref{r4}} and {\eqref{r5}} ensure acceptance of all words in $(\Sigma_{n-1}^* a_n)^{i} L(\A_{k-i,n-1})a_n \Sigma_n^*$, for which the longest factor before the $(i+1)$th occurrence of $a_n$ is not of the form $W_{k-i,n-1}$, and hence is not a correct factor of $W_{k,n} = W_{k,n-1} a_n \cdots a_n W_{k-i,n-1} a_n \cdots a_n W_{1,n-1} a_n$.
    Together, these transitions ensure that $\A_{k,n}$ accepts every input other than $W_{k,n}$. 
    Notice that transitions~\eqref{r6} are not needed to obtain the language, and we show below that they do not change the language. This corresponds to the informal explanation that the lower part of the automaton is used to take care of confluence and has no effect on the language. We formulate this precisely as Corollary~\ref{cor_nonacc}. 
    
    It remains to show that $\A_{k,n}$ does not accept $W_{k,n}$, which we do by induction on $(k,n)$. We start with the base cases. For $(0,n)$ and any $n\geq 1$, the word $W_{0,n}=\eps$ is not accepted by $\A_{0,n}$, since the initial states $(0;m)=(k;m)$ of $\A_{0,n}$ are not accepting. Likewise, for $(k,1)$ and any $k\ge 0$, we find that $W_{k,1}=a_1^k$ is not accepted by $\A_{k,1}$ (\cf\ Figure~\ref{ptnfa3}).

    For the inductive case $(k,n)\ge (1,2)$, assume that $\A_{k',n'}$ does not accept $W_{k',n'}$ for any $(k',n') < (k,n)$. We have $W_{k,n} = W_{k,n-1} a_n W_{k-1,n}$, and $W_{k,n-1}$ is not accepted by $\A_{k,n-1}$ by induction. Therefore, after reading $W_{k,n-1} a_n$, automaton $\A_{k,n}$ must be in one of the states $(1;m)$, $1\le m\le n$, or $(k+1;n)$. However, states $(1;m)$, $1\le m\le n$, are the initial states of $\A_{k-1,n}$, which does not accept $W_{k-1,n}$ by induction, and hence neither $\A_{k,n}$ accepts $W_{k-1,n}$ from the states $(1;m)$, $1\le m\le n$; indeed, after reading $W_{k-1,n}$, automaton $\A_{k-1,n}$ ends up in non-accepting states, and since $\A_{k,n}$ starting from the states $(1;m)$, $1\le m\le n$, differs from $\A_{k-1,n}$ only in states $(2k;i)$, $i=1,\ldots,n$, it ends up in the same non-accepting states after reading $W_{k-1,n}$ as $\A_{k-1,n}$. Thus, assume that $\A_{k,n}$ is in state $(k+1;n)$ after reading $W_{k,n-1} a_n$. Since $W_{k-1,n}$ has exactly $k-1$ occurrences of letter $a_n$, $\A_{k,n}$ is in state $(2k;n)$ after reading $W_{k-1,n}$. Hence $W_{k,n}$ is not accepted by $\A_{k,n}$.
  \end{proof}
    
  The last part of the previous proof shows that the suffix $W_{k-1,n}$ of the word $W_{k,n}=W_{k,n-1}a_nW_{k-1,n}$ is not accepted from state $(k+1;n)$. This can be generalized as follows. 
  \begin{cor}\label{WnnStructure}
    For any suffix $a_i w$ of $W_{k,n}$, $w$ is not accepted from state $(k+1;i)$ of $\A_{k,n}$.
  \end{cor}
  \begin{proof}
    Consider the word $W_{k,n}$ over $\Sigma_n=\{a_1,a_2,\ldots,a_n\}$ constructed in the proof of Lemma~\ref{exprponfas}, and let $i\in\{1,\ldots,n\}$ be the maximal number for which there is a suffix $a_i w$ of $W_{k,n}$ such that $w$ is accepted by $\A_{k,n}$ from state $(k+1;i)$. Then $W_{k,n}=w_1a_i w_2 w_3$, where $w_2\in\{a_1,\ldots,a_i\}^*$ is the shortest word labeling the path from state $(k+1;i)$ to state $max$. By the construction of $\A_{k,n}$, word $a_iw_2$ must contain $k+1$ letters $a_i$. We shown that $W_{k,n}$ does not contain more than $k$ letters $a_i$ interleaved only with letters $a_j$ for $j<i$, which yields a contradiction that proves the claim.
      
    By definition, every longest factor of $W_{k,n}$ over $\{a_1,\ldots,a_i\}$ is of the form $W_{k-\ell,i}$, for $\ell \in\{0,\ldots,k-1\}$. Since $W_{k-\ell,i} = W_{k-\ell,i-1}\, a_i\, W_{k-\ell-1,i-1}\, a_i\, \cdots\, a_i\, W_{1,i-1}\, a_i$, the number of occurrences of $a_i$ interleaved only with letters $a_j$ for $j<i$ is at most $k-\ell$, which results in the maximum of $k$ for $\ell=0$ as claimed above.
  \end{proof}

  As already pointed out in the proof of Lemma~\ref{exprponfas}, transitions~\eqref{r6} are redundant and present only to take care of confluence.
  \begin{cor}\label{cor_nonacc}
    Removing from $\A_{k,n}$ the non-accepting states $(k+1,i),\ldots,(2k,i)$, for $1\le i \le n$, and the corresponding transitions results in an rpoNFA that accepts the same language.
  \end{cor}
  \begin{proof}
    By the proof of Lemma~\ref{exprponfas}, removing the states with corresponding transitions has no effect on the accepted language. The resulting automaton is indeed an rpoNFA. This rpoNFA is exactly the rpoNFA used in our previous work~\cite{mfcs16:mktmmt_full}. 
  \end{proof}

  We now have the necessary results to show, using a reduction from the word problem of polynomially-space-bounded Turing machines, that the universality problem for ptNFAs, where the alphabet may grow polynomially with the number of states, is \PSpace-complete.
  
  A {\em deterministic Turing machine} (DTM) is a tuple $M = (Q,T,I,\delta,\blank,q_o,q_f)$, where $Q$ is the finite state set, $T$ is the tape alphabet, $I\subseteq T$ is the input alphabet, $\blank\in T \setminus I$ is the blank symbol, $q_o$ is the initial state, $q_f$ is the accepting state, and $\delta$ is the transition function mapping $Q\times T$ to $Q\times T \times \{L,R,S\}$;
  see Aho et al.~\cite{AhoHU74} for details.

  \begin{thm}\label{thmMain}
    Deciding universality for ptNFAs is \PSpace-complete.
  \end{thm}
  \begin{proof}
    Membership follows since universality is in \PSpace for NFAs~\cite{GareyJ79}. 
  
    To prove \PSpace-hardness, we consider a polynomial $p$ and a $p$-space-bounded DTM $\M = (Q,T,I,\delta,\blank,q_o,q_f)$ with a nonempty input $x$. The basic idea of the proof is to use the alphabet $\Pi = \Sigma_n \times \Delta$, where $\Sigma_n=\{a_1,\ldots,a_n\}$ and $\Delta$ is used to encode runs of the Turing machine $\M$. The aim is to construct a ptNFA that accepts all words over $\Pi$, where the projection to the first component does not equal $W_{n,n}$ (the word constructed in Lemma~\ref{exprponfas}) or the projection to the second component does not encode an accepting run of $\M$ on $x$. 
    
    Before elaborating on the details of the construction, we make several assumptions on the DTMs that simplify the proof and under which the word problem for these DTMs clearly remains \PSpace-hard. We assume that
    \begin{enumerate}
      \item the initial and accepting states of $\M$ are different, \ie, $q_o\neq q_f$;
      
      \item $\M$ accepts by looping in the accepting state $q_f$ indefinitely, \ie, no transition from state $q_f$ modifies the tape, state, or head position;

     \item and $\M$ always accepts with the head at the very beginning of the tape.
    \end{enumerate}

    A configuration of $\M$ on $x$ consists of a current state $q\in Q$, the position $1\leq \ell\leq p(|x|)$ of the read/write head, and the tape contents $\theta_1,\ldots,\theta_{p(|x|)}$ with $\theta_i\in T$. We represent it by a sequence 
    \[
      \tuple{\theta_1,\epsnostate}\cdots\tuple{\theta_{\ell-1},\epsnostate}\tuple{\theta_{\ell},q}\tuple{\theta_{\ell+1},\epsnostate}\cdots\tuple{\theta_{p(|x|)},\epsnostate}
    \]
    of symbols from $\Delta = T\times(Q\cup\{\epsnostate\})$. A run of $\M$ on $x$ is represented as a word $\# w_1 \# w_2 \# \cdots \allowbreak \# w_m \#$, where $w_i\in\Delta^{p(|x|)}$ and $\#\notin\Delta$ is a fresh separator symbol. One can construct a regular expression recognizing all words over $\Delta\cup\{\#\}$ that do not correctly encode a run of $\M$ (in particular are not of the form $\# w_1 \# w_2 \# \cdots \allowbreak \# w_m \#$) or that encode a run that is not accepting~\cite{AhoHU74}. 
    Such a regular expression can be constructed in the following three steps: we detect all words that
    \begin{description}
      \item[(A)] do not start with the initial configuration; 
      
      \item[(B)] do not encode a valid run since they violate a transition rule (including words with an invalid encoding);

      \item[(C)] encode non-accepting runs or runs that end prematurely.
    \end{description}

    If $\M$ has an accepting run on $x$, it has one without repeated configurations. There are $C(x) = |\Delta|^{p(|x|)}$ distinct configuration words in our encoding. Considering a separator symbol $\#$, the length of the encoding of a run without repeated configurations is at most $1+ C(x)(p(|x|)+1)$, because every configuration word ends with $\#$ and is thus of length $p(|x|)+1$. Let $n$ be the least number such that $|W_{n,n}|\geq 1+ C(x)(p(|x|)+1)$, where $W_{n,n}$ is the word constructed in Lemma~\ref{exprponfas}. Since $|W_{n,n}|+1=\binom{2n}{n} \ge 2^n$, it follows that $n$ is smaller than $\lceil\log(1+ C(x)(p(|x|)+1))\rceil$, and hence polynomial in the size of $\M$ and $x$.

    Consider the ptNFA $\A_{n,n}$ over the alphabet $\Sigma_n=\{a_1,\ldots,a_n\}$ of Lemma~\ref{exprponfas}, and define the alphabet $\Deltaplus = \Delta \cup\{\#,\$\}$. We consider the alphabet 
    $ 
      \Pi=\Sigma_n\times\Deltaplus
    $ 
    where the first component is an input for $\A_{n,n}$ and the second component is used for encoding a run as described above; that is, letters of $\Pi$ are pairs that might have pairs (of tape symbols and states) as their second element. Recall that $\A_{n,n}$ accepts all words different from $W_{n,n}$. Therefore, only those words over $\Pi$ are of our interest, where the first components form the word $W_{n,n}$. Since the length of $W_{n,n}$ may not be a multiple of $p(|x|)+1$, we add $\$$ to fill up any remaining space after the last configuration.

    For a word $w=\tuple{a_{i_1},\delta_1}\cdots \tuple{a_{i_\ell},\delta_\ell}\in\Pi^\ell$, we define $w[1]=a_{i_1}\cdots a_{i_\ell} \in \Sigma_n^\ell$ as the projection of $w$ to the first components, and $w[2]=\delta_1\ldots\delta_\ell\in\Deltaplus^\ell$ as the projection to the second components. Conversely, for a word $a_{i_1}\cdots a_{i_\ell} \in \Sigma_n^\ell$, we write $a_{i_1}\cdots a_{i_\ell}\otimes \Deltaplus$ to denote the set $(\{a_{i_1}\}{\times}\Deltaplus)\cdots (\{a_{i_\ell}\}{\times}\Deltaplus)$ of all words $w\in\Pi^{\ell}$ with $w[1]=a_{i_1}\cdots a_{i_\ell}$. Similarly, for $\delta_1\ldots\delta_\ell\in\Deltaplus^\ell$, we write $\Sigma_n\otimes \delta_1\ldots\delta_\ell$ to denote the set $(\Sigma_n{\times}\{\delta_1\})\cdots (\Sigma_n{\times}\{\delta_\ell\})$ of all words $w\in\Pi^{\ell}$ with $w[2]=\delta_1\ldots\delta_\ell$. We extend this notation to sets of words.

    Let $\enc(\A_{n,n})$ denote the automaton $\A_{n,n}$ with each transition $q\xrightarrow{a_i} q'$ replaced by all transitions $q\xrightarrow{\pi} q'$ with $\pi \in \{a_i\}\times \Deltaplus$. Then $\enc(\A_{n,n})$ accepts the language $\Pi^*\setminus (W_{n,n}\otimes \Deltaplus)$. We say that a word $w$ encodes an accepting run of $\M$ on $x$ if
    $w[1]=W_{n,n}$ and
    $w[2]$ is of the form $\#w_1\#\cdots\# w_m \# \$^j$ 
    such that there is an $i\in\{1,2,\ldots,m\}$ for which
    $\#w_1\#\cdots\#w_i\#$ encodes an accepting run of $\M$ on $x$, 
    $w_k=w_i$ for all $k\in\{i+1,\ldots,m\}$, and
    $j\leq p(|x|)$.
    That is, we extend the encoding by repeating the accepting configuration until we have less than $p(|x|)+1$ symbols before the end of $|W_{n,n}|$, and fill up the remaining places with symbol $\$$. This extension is possible due to the assuption that $\M$ loops in the accepting configuration.

    For {\bf (A)}, we want to detect all words that do not start with the word
    \begin{align}
      w[2] =\#\tuple{x_1,q_0}\allowbreak\tuple{x_2,\epsnostate}\cdots\allowbreak\tuple{x_{|x|},\epsnostate}\tuple{\blank,\epsnostate}\cdots\tuple{\blank,\epsnostate}\#
      \label{eq_correct_initial_word}
    \end{align}
    of length $p(|x|)+2$.
    This happens if 
      {\bf (A.1)} the word is shorter than $p(|x|)+2$, or 
      {\bf (A.2)} at position $j$, for $0\leq j\leq p(|x|)+1$, there is a letter from the alphabet $\Deltaplus \setminus\{\delta_j\}$, where $\delta_j$ is the $j$th letter of the expected initial word \eqref{eq_correct_initial_word}.
    Let $\bar{E}_j=\Sigma_n \times (\Deltaplus\setminus\{\delta_j\})$. We can capture {\bf (A.1)} and {\bf (A.2)} in the regular expression
    \begin{equation}\label{eq_re_wrong_start}
      \left(\varepsilon + \Pi + \Pi^2 +\cdots+ \Pi^{p(|x|)+1}\right) 
      + \sum_{0\leq j\leq p(|x|)+1} (\Pi^{j} \cdot \bar{E}_j\cdot\Pi^*)
    \end{equation}
    where $\Pi^k$ is an abbreviation of the concatenation of $\Pi$ $k$-times.

    Expression \eqref{eq_re_wrong_start} is polynomial in size. It can be captured by a ptNFA as follows. Each of the first $p(|x|)+2$ expressions defines a finite language and can easily be captured by a ptNFA (by a confluent DFA) of size of the expression. (By size of a regular expression, we mean the ordinary length, \ie, the total number of symbols, including parentheses; \cf~Ellul et al.~\cite{EllulKSW05} for more options and a detailed discussion.) The disjoint union of these ptNFAs then clearly forms a single ptNFA recognizing the language $\varepsilon + \Pi + \Pi^2 +\cdots+ \Pi^{p(|x|)+1}$.

    To express the language $\Pi^{j} \cdot \bar{E}_j\cdot\Pi^*$ as a ptNFA, we first construct the minimal incomplete DFA recognizing this language (states $0,1,\ldots,j,max$ in Figure~\ref{fig_const_1}). However, we cannot complete this DFA by simply adding the missing transitions under $\Sigma_n\times \{\delta_j\}$ from state $j$ to a new sink state because it results in a DFA with two maximal states -- $max$ and the sink state -- violating the UMS property. Instead, we use a copy of the ptNFA $\enc(\A_{n,n})$ and add the missing transitions from state $j$ under $\tuple{a_i,\delta_j}\in \Sigma_n\times \{\delta_j\}$ to state $(n+1;i)$; see Figure~\ref{fig_const_1} for an illustration. Notice that states $(n+1;i)$ are the states $(k+1;i)$ in Figure~\ref{ptnfa3}.
    \begin{figure}
      \centering
      \includegraphics[angle=0,scale=.7]{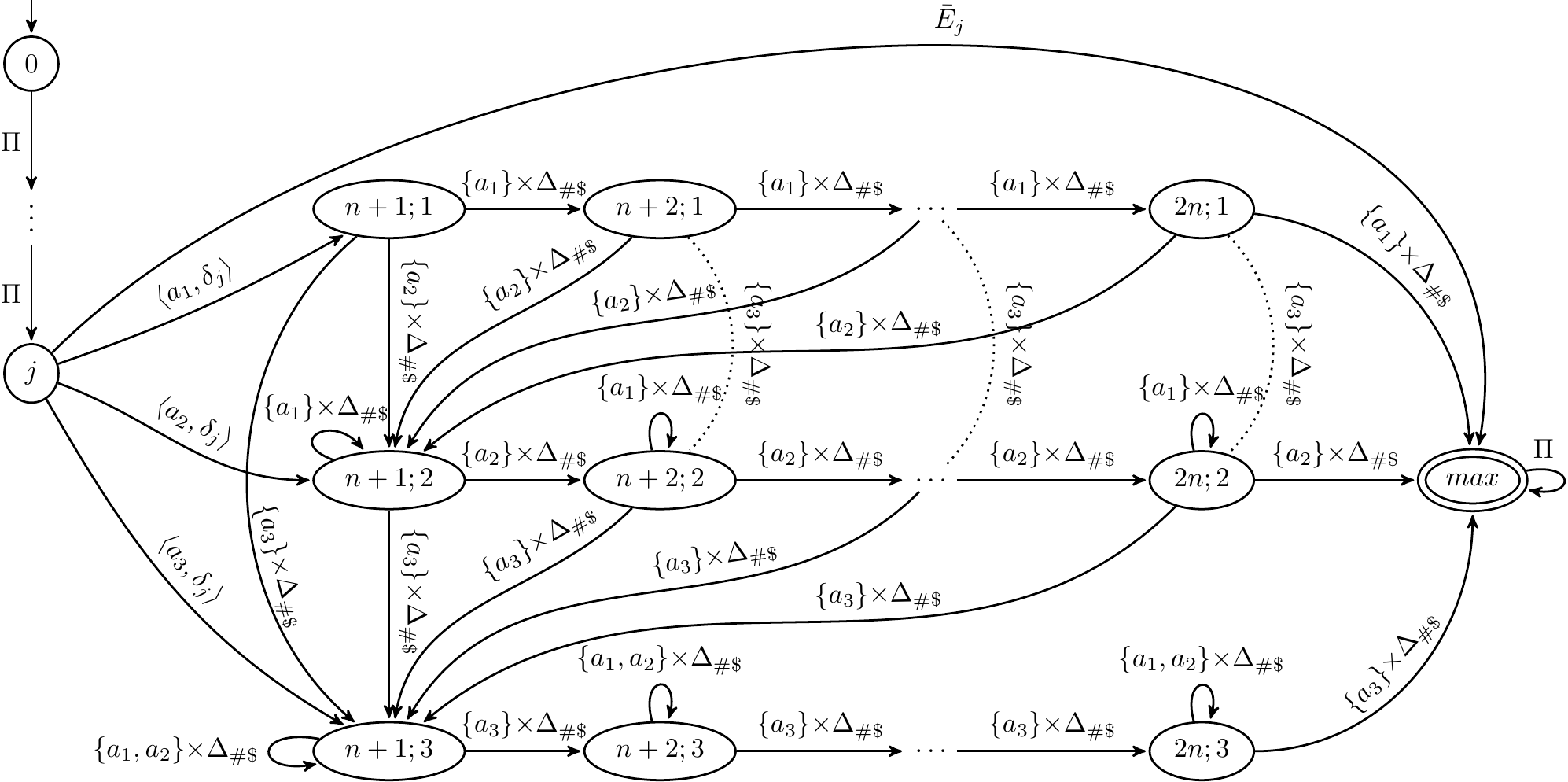}
      \caption{A ptNFA that accepts $\Pi^{j} \cdot \bar{E}_j\cdot\Pi^* + (\Pi^*\setminus (W_{n,n}\otimes \Deltaplus))$ with $\bar{E}_j=\Sigma_n\times (\Deltaplus\setminus\{\delta_j\})$ illustrated for $\Sigma_n=\{a_1,a_2,a_3\}$; only the relevant part of $\A_{n,n}$ is depicted.}
      \label{fig_const_1}
    \end{figure}
    The resulting automaton is a ptNFA, since it is complete, partially ordered, and satisfies the UMS property -- for every state $q$ different from {\em max}, the component co-reachable and reachable under the letters of self-loops in state $q$ is only state $q$ itself. This ptNFA accepts all words of $\Pi^{j} \cdot \bar{E}_j\cdot\Pi^*$.
    
    We now show that any word $w$ that is accepted by this ptNFA and that does not belong to $\Pi^{j} \cdot \bar{E}_j\cdot\Pi^*$ is such that $w[1] \neq W_{n,n}$, \ie, $w$ belongs to $\Pi^*\setminus (W_{n,n}\otimes\Deltaplus)$.
    To this aim, assume that $w[1]=W_{n,n}$ and that $w$ is of the form $w=u \tuple{a_i,\delta_j} v$ such that
    $|u|=j$. Then, $\tuple{a_i,\delta_j}$ is the letter under which the state $(n+1;i)$ of $\A_{n,n}$ is reached and $v$ is read in the $\A_{n,n}$-part of the ptNFA. By Corollary~\ref{WnnStructure}, $v$ is not accepted from state $(n+1;i)$, and hence the ptNFA does not accept $w$. Therefore, the ptNFA accepts the language $\Pi^{j} \cdot \bar{E}_j\cdot\Pi^* + (\Pi^*\setminus (W_{n,n}\otimes\Deltaplus))$. Constructing such a ptNFA for polynomially many expressions $\Pi^{j} \cdot \bar{E}_j \cdot \Pi^*$ and taking their disjoint union results in a polynomially large ptNFA accepting the language $\sum_{j=0}^{p(|x|)+1} (\Pi^{j} \cdot \bar{E}_j\cdot\Pi^*) + (\Pi^*\setminus (W_{n,n}\otimes\Deltaplus))$.

    Notice that we ensure that the surrounding $\#$ in the initial configuration are present.

    For {\bf (B)}, we check for incorrect transitions and invalid encodings. Consider again the encoding $\#w_1\#\cdots \allowbreak \#w_m\#$ of a sequence of configurations with a word over $\Delta\cup\{\#\}$. We can assume that $w_1$ encodes the initial configuration according to {\bf (A)}. In an encoding of a valid run, the symbol at any position $j\geq p(|x|)+2$ is uniquely determined by the three symbols at positions $j-p(|x|)-2$, $j-p(|x|)-1$, and $j-p(|x|)$, corresponding to the cell and its left and right neighbor in the previous configuration. Given symbols $\delta_\ell,\delta,\delta_r\in\Delta\cup\{\#\}$, we can therefore define $f(\delta_\ell,\delta,\delta_r)\in\Delta\cup\{\#\}$ to be the symbol required in the next configuration. The case where $\delta_\ell=\#$ or $\delta_r=\#$ corresponds to transitions applied at the left and right edge of the tape, respectively; for the cases where $\delta=\#$, we define $f(\delta_\ell,\delta,\delta_r)=\#$, ensuring that the separator $\#$ is always present in successor configurations as well. We extend $f$ to $f\colon \Deltaplus^3\to \Deltaplus$. We can then check for invalid transitions using the regular expression
    \begin{equation}\label{eq_re_wrong_trans}
      \Pi^*\, \sum_{\delta_\ell,\delta,\delta_r \in \Deltaplus} (\Sigma_n\otimes \delta_\ell\delta\delta_r)\cdot  \Pi^{p(|x|)-1} \cdot \hat{f}(\delta_\ell,\delta,\delta_r)\cdot \Pi^*
    \end{equation}
    where $\hat{f}(\delta_\ell,\delta,\delta_r)$ is $\Pi \setminus (\Sigma_n \times \{f(\delta_\ell,\delta,\delta_r), \$\})$. Here, we also consider $\$$ as a correct continuation instead of the expected next configuration symbol.
    Note that \eqref{eq_re_wrong_trans} only detects wrong transitions and invalid encodings if a long enough next configuration exists. The case that the run stops prematurely is covered in {\bf (C)}.

    Expression \eqref{eq_re_wrong_trans} is not readily encoded in a ptNFA because of the leading $\Pi^*$. To address this issue, we replace $\Pi^*$ by the expression $\Pi^{\leq |W_{n,n}|-1}$, which matches every word $w\in\Pi^*$ with $|w|\leq |W_{n,n}|-1$. This suffices for our purpose because the computations of interest are of length $|W_{n,n}|$ and a violation of a correct computation must occur.
    As $|W_{n,n}|-1$ is exponential, we cannot encode it directly and we use $\enc(\A_{n,n})$ instead.

    \begin{figure}
      \centering
      \includegraphics[angle=0,scale=.7]{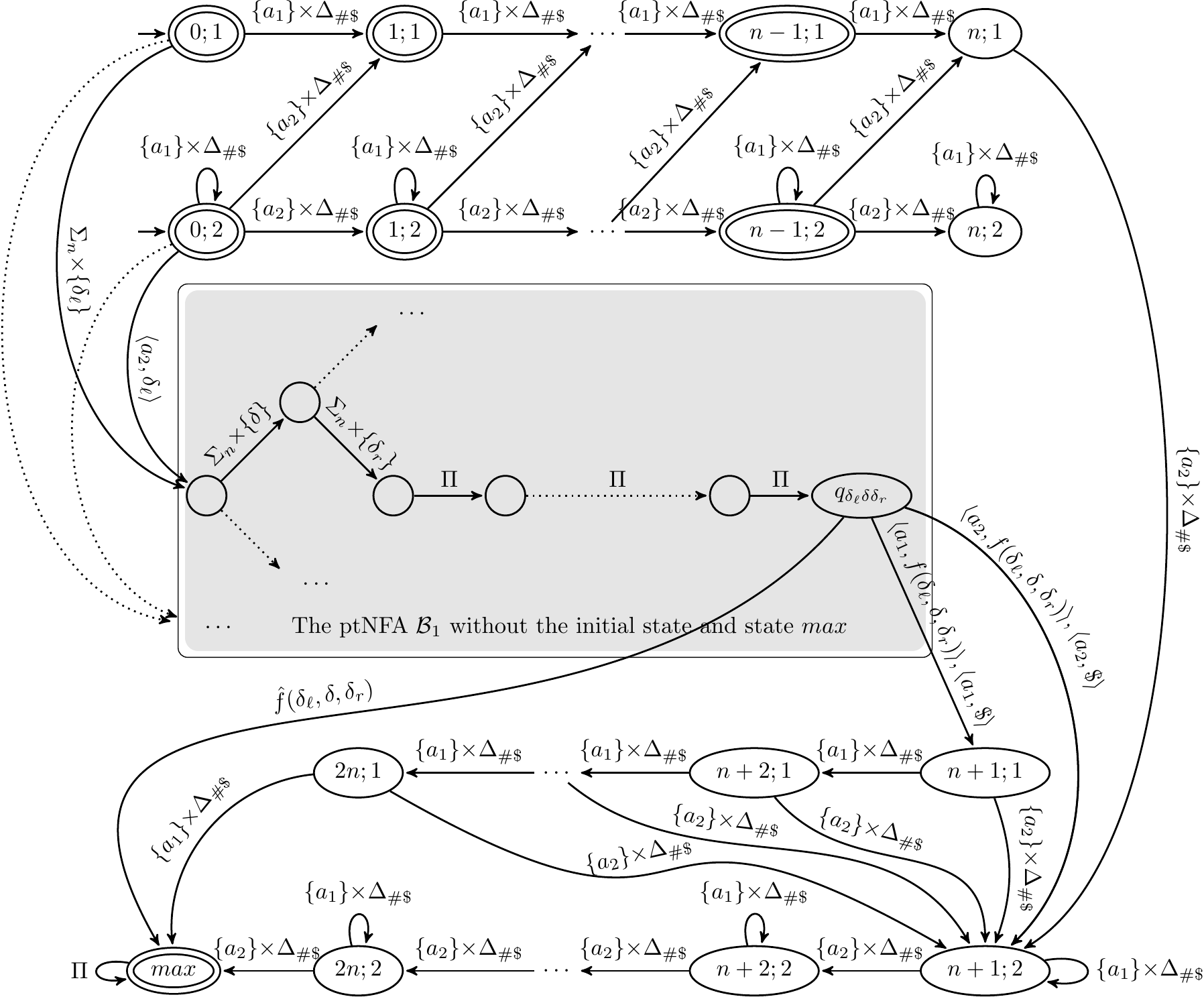}
      \caption{The ptNFA $\B$ that combines ptNFA $\enc(\A_{n,n})$ with ptNFA $\B_1$, for $n=2$; the edges from accepting states of $\enc(\A_{n,n})$ to the second states of $\mathcal{B}_1$ are illustrated only for states $(0;1)$ and $(0;2)$.}
      \label{fig_tree}
    \end{figure}

    In detail, let $E$ be the expression obtained from \eqref{eq_re_wrong_trans} by omitting the initial $\Pi^*$, and let $\B_1$ be an incomplete DFA that accepts the language of $E$ constructed as follows. From the initial state, we construct a tree-shaped DFA corresponding to all words of length three of the finite language $\sum_{\delta_\ell,\delta,\delta_r \in \Deltaplus} (\Sigma_n\otimes \delta_\ell \delta \delta_r)$. To every leaf state, we add a path under $\Pi$ of length $p(|x|)-1$. The result corresponds to the language $\sum_{\delta_\ell,\delta,\delta_r \in \Deltaplus} (\Sigma_n\otimes \delta_\ell \delta \delta_r) \cdot \Pi^{p(|x|)-1}$. Let $q_{\delta_\ell \delta \delta_r}$ denote the states uniquely determined by the words in $(\Sigma_n\otimes \delta_\ell \delta \delta_r) \cdot \Pi^{p(|x|)-1}$. We add the transitions $q_{\delta_\ell \delta \delta_r} \xrightarrow{\hat{f}(\delta_\ell,\delta,\delta_r)} max$, where $max$ is the state of $\enc(\A_{n,n})$. Automaton $\B_1$ is illustrated in the middle part of Figure~\ref{fig_tree}, except for the initial state that is identified with accepting states of the ptNFA $\enc(\A_{n,n})$ as described below, and state $max$ that is a state of $\enc(\A_{n,n})$. It is an incomplete DFA for the language of $E$ of polynomial size. It is incomplete only in states $q_{\delta_r\delta\delta_\ell}$ due to the missing transitions under $\Sigma_n\times \{f(\delta_\ell,\delta,\delta_r)\}$ and $\Sigma_n\times \{\$\}$. We complete it by adding the missing transitions to the states of the ptNFA $\enc(\A_{n,n})$. Namely, for $z\in \{ \tuple{a_i,f(\delta_\ell,\delta,\delta_r)},\tuple{a_i,\$}\}$, we add $q_{\delta_\ell \delta \delta_r} \xrightarrow{~z~} (n+1;i)$; see Figure~\ref{fig_tree} for an illustration.

    We construct a ptNFA $\B$ accepting the language $(\Pi^*\setminus (W_{n,n}\otimes \Deltaplus)) + (\Pi^{\leq |W_{n,n}|-1}\cdot E)$ by merging $\enc(\A_{n,n})$ with the DFA $\B_1$, where we add edges labeled with $(\Sigma_n \setminus \Sigma(q)) \times \{\delta_{\ell}\}$ from any accepting state $q$ of $\enc(\A_{n,n})$ to the states of the second level of the tree-shaped DFA $\B_1$ (the leftmost states in the middle part of Figure~\ref{fig_tree}). This step is justified by Corollary~\ref{cor_nonacc}, since we do not need to consider connecting $\B_1$ to non-accepting states of $\enc(\A_{n,n})$ and it is not possible to connect it to state $max$. The fact that $\enc(\A_{n,n})$ alone accepts $\Pi^*\setminus (W_{n,n}\otimes\Deltaplus)$ was shown in Lemma~\ref{exprponfas}. This also implies that it accepts all words of length $\leq |W_{n,n}|-1$ as needed to show that $(\Pi^{\leq |W_{n,n}|-1}\cdot E)$ is accepted. Entering states of $\B_1$ after reading a word of length $\geq|W_{n,n}|$ is possible but all such words are longer than $W_{n,n}$, and hence in $\Pi^*\setminus (W_{n,n}\otimes\Deltaplus)$.

    To show that the completion does not affect the language, let $w$ be a word that is read but not accepted by $\B_1$, and let $u$ lead $\enc(\A_{n,n})$ to a state from which $w$ is read in $\B_1$. Since $w$ is not accepted, there is a letter $z$ and a word $v$ such that $uwz$ goes to state $(n+1;i)$ of $\enc(\A_{n,n})$ (for $z[1]=a_i$) and $v$ leads $\enc(\A_{n,n})$ from state $(n+1;i)$ to state $max$. If $u[1]w[1]a_iv[1] = W_{n,n,}$, then $v$ is not accepted from $(n+1;i)$ by Corollary~\ref{WnnStructure}, and hence $uwzv[1]\neq W_{n,n}$; thus, $uwzv\notin (W_{n,n}\otimes\Deltaplus)$.

    It remains to show that for every proper prefix $w_{n,n}$ of $W_{n,n}$, there is a state in $\A_{n,n}$ reached by $w_{n,n}$ that has transitions to the second states of $\B_1$, and hence the check represented by $E$ in $\Pi^{\leq |W_{n,n}|-1}\cdot E$ can be performed. In other words, if $a_{n,n}$ denotes the letter following $w_{n,n}$ in $W_{n,n}$, then there must be a state reachable by $w_{n,n}$ in $\A_{n,n}$ that does not have a self-loop under $a_{n,n}$. However, this follows from the fact that $\A_{n,n}$ accepts everything but $W_{n,n}$, since then the DFA obtained from $\A_{n,n}$ by the standard subset construction has a path of length $\binom{2n}{n}-1$ labeled with $W_{n,n}$ without any loop. Moreover, any state of this path in the DFA is a subset of states of $\A_{n,n}$, and therefore at least one of the states reachable under $w_{n,n}$ in $\A_{n,n}$ does not have a self-loop under $a_{n,n}$.

    The ptNFA $\B$ thus accepts the language $(\Pi^{\leq |W_{n,n}|-1}\cdot E) + (\Pi^*\setminus (W_{n,n}\otimes \Deltaplus))$.
    
    Finally, for {\bf (C)}, we detect all words that 
      {\bf(C.1)} end in a configuration that is incomplete (too short), possibly followed by at most $p(|x|)$ trailing $\$$,
      {\bf(C.2)} end in a configuration that is not in the accepting state $q_f$, which must be the very first tape symbol by our assumption,
      {\bf(C.3)} end with more than $p(|x|)$ trailing $\$$, or
      {\bf(C.4)} contain $\$$ not only at the last positions, that is, we detect all words where $\$$ is followed by a different symbol.
    For a word $v$, we use $v^{\leq i}$ to abbreviate $\varepsilon + v + \cdots + v^i$, and we define $\bar{E}_f= (T\times (Q\setminus\{q_f\}))$. Then these properties are expressed by the following expressions:    
    \begin{align}
    \begin{array}{@{}l@{~~}l@{}}
    \textbf{(C.1)} & \Pi^*\, (\Sigma_n\times \{\#\})\, (\Pi + \cdots + \Pi^{p(|x|)})\, (\Sigma_n \times \{\$\})^{\leq p(|x|)} + {} \\
    \textbf{(C.2)} & \Pi^*\, (\Sigma_n \times \bar{E}_f)\, \Pi^{p(|x|)-1}\, (\Sigma_n \times \{\#\})\, (\Sigma_n \times \{\$\})^{\leq p(|x|)} + {} \\
    \textbf{(C.3)} & \Pi^*\, (\Sigma_n \times \{\$\})^{p(|x|)+1} + {} \\
    \textbf{(C.4)} & (\Pi\setminus (\Sigma_n\times \{\$\}))^*\, (\Sigma_n \times \{\$\})\, (\Sigma_n\times \{\$\})^*\, (\Pi\setminus (\Sigma_n\times \{\$\}))\, \Pi^*
    \end{array}\hspace{2cm}\label{eq_re_wrong_final}
    \end{align}

    As before, we cannot encode the expression directly as a ptNFA, but we can perform a similar construction as the one used for encoding \eqref{eq_re_wrong_trans}. Namely, ptNFAs for {\bf (C.1)} and {\bf (C.2)} are illustrated in Figure~\ref{figC2} and for {\bf (C.3)} in Figure~\ref{figC3}, where, analogously to Figure~\ref{fig_tree}, the edges from the accepting states of $\enc(\A_{n,n})$ to the second state of the automaton $\mathcal{C}_i$ that recognizes the language of expression {\bf (C.i)} without the initial $\Pi^*$, for $i=1,2,3$, are illustrated only for states $(0;1)$ and $(0;2)$. However, the reader should keep in mind that such transitions lead from every accepting state of $\enc(\A_{n,n})$ to the second state of $\mathcal{C}_i$. Finally, {\bf (C.4)} can be represented by a three-state partially ordered and confluent DFA.

    The expressions \eqref{eq_re_wrong_start}--\eqref{eq_re_wrong_final} together then detect all non-accepting or wrongly encoded runs of $\M$. In particular, if we start from the correct initial configuration (\eqref{eq_re_wrong_start} does not match), then for \eqref{eq_re_wrong_trans} not to match, all complete future configurations must have exactly one state, be delimited by encodings of $\#$, and correctly follow from the previous configurations. We have shown how to express each of the expressions as a ptNFA. Taking the disjoint union of all these ptNFAs results in a single ptNFA. This ptNFA is of polynomial size, and hence we have reduced the word problem of polynomially-space-bounded Turing machines to the universality problem for ptNFAs.
  \end{proof}    
  
  \begin{figure}
    \centering
    \includegraphics[angle=0,scale=.7]{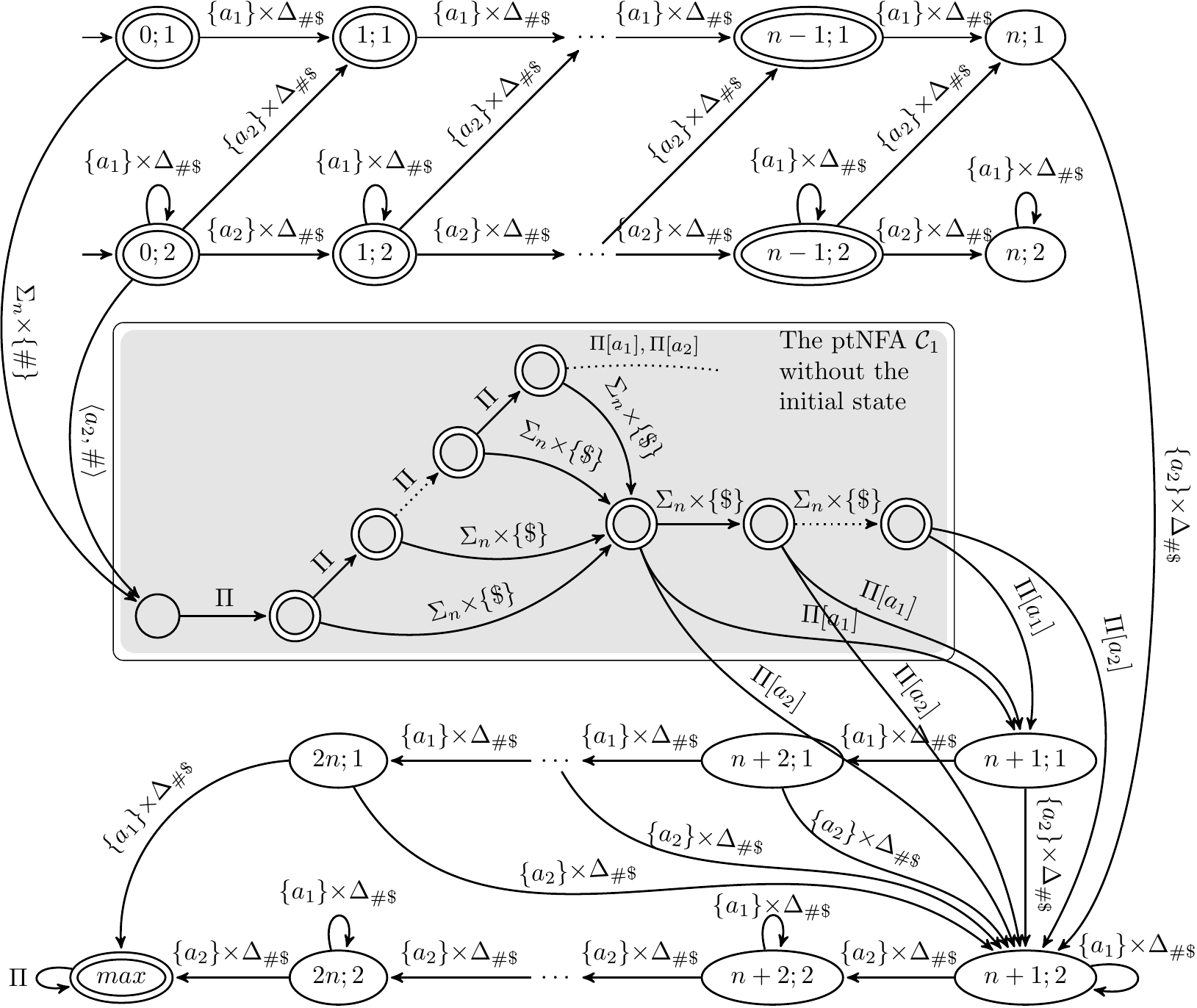}
    
    \vspace{.7cm}
    
    \includegraphics[angle=0,scale=.7]{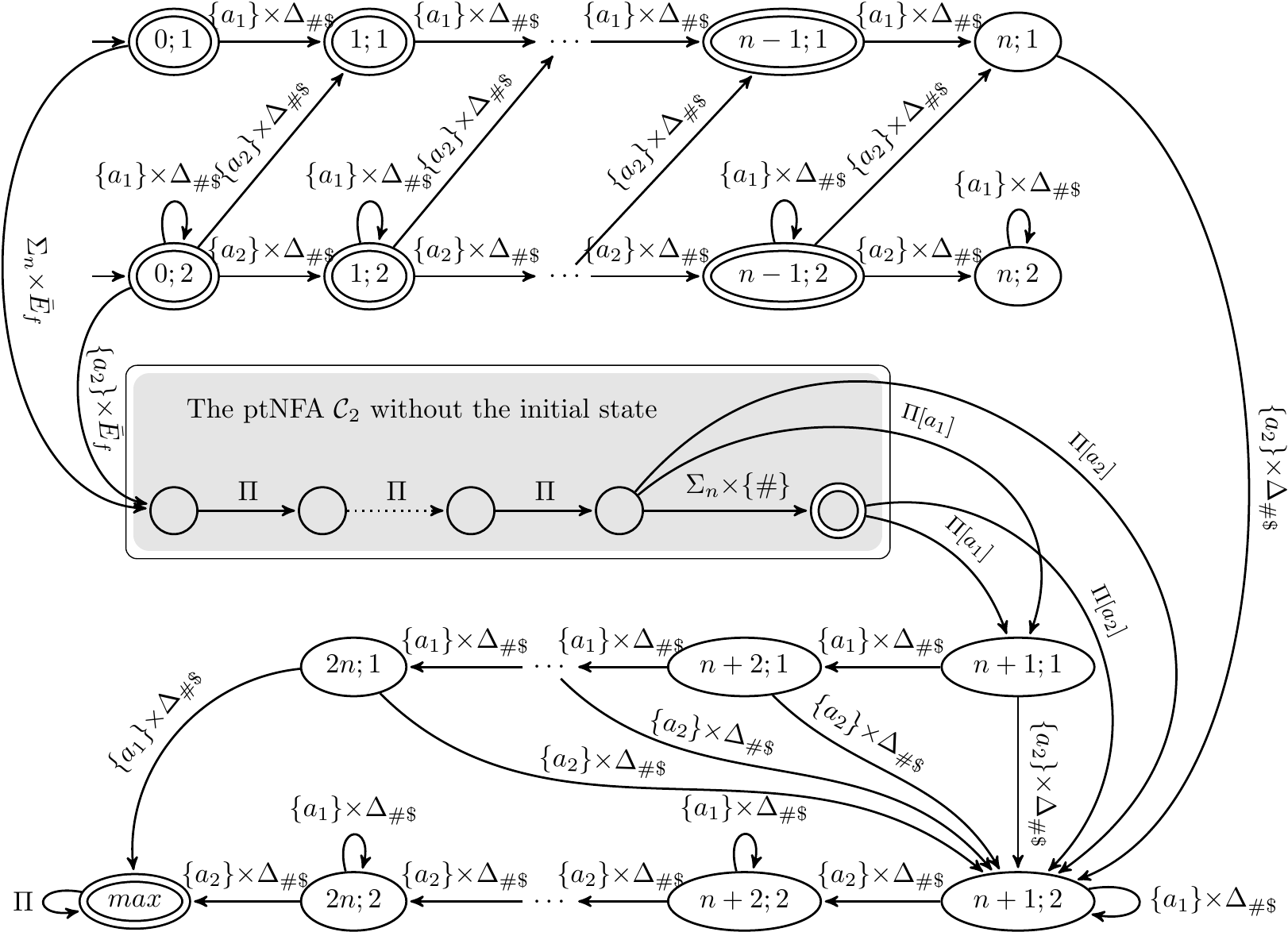}
    \vspace{0cm}
    \caption{ptNFAs for {\bf (C.1)} (top) and {\bf (C.2)} (bottom) illustrated for $n=2$; automata $\mathcal{C}_1$ and $\mathcal{C}_2$, recognizing resp.\! {\bf (C.1)} and {\bf (C.2)} without the initial $\Pi^*$, are completed by adding transitions $q\xrightarrow{\Pi[a_i]} (n+1;i)$, where $\Pi[a_i]$ denotes all letters of $\Pi$ undefined in state $q$ with the first component of the letter being $a_i$.}
    \label{figC2}\label{figC1}
  \end{figure}

  \begin{figure}
    \centering
    \includegraphics[angle=0,scale=.7]{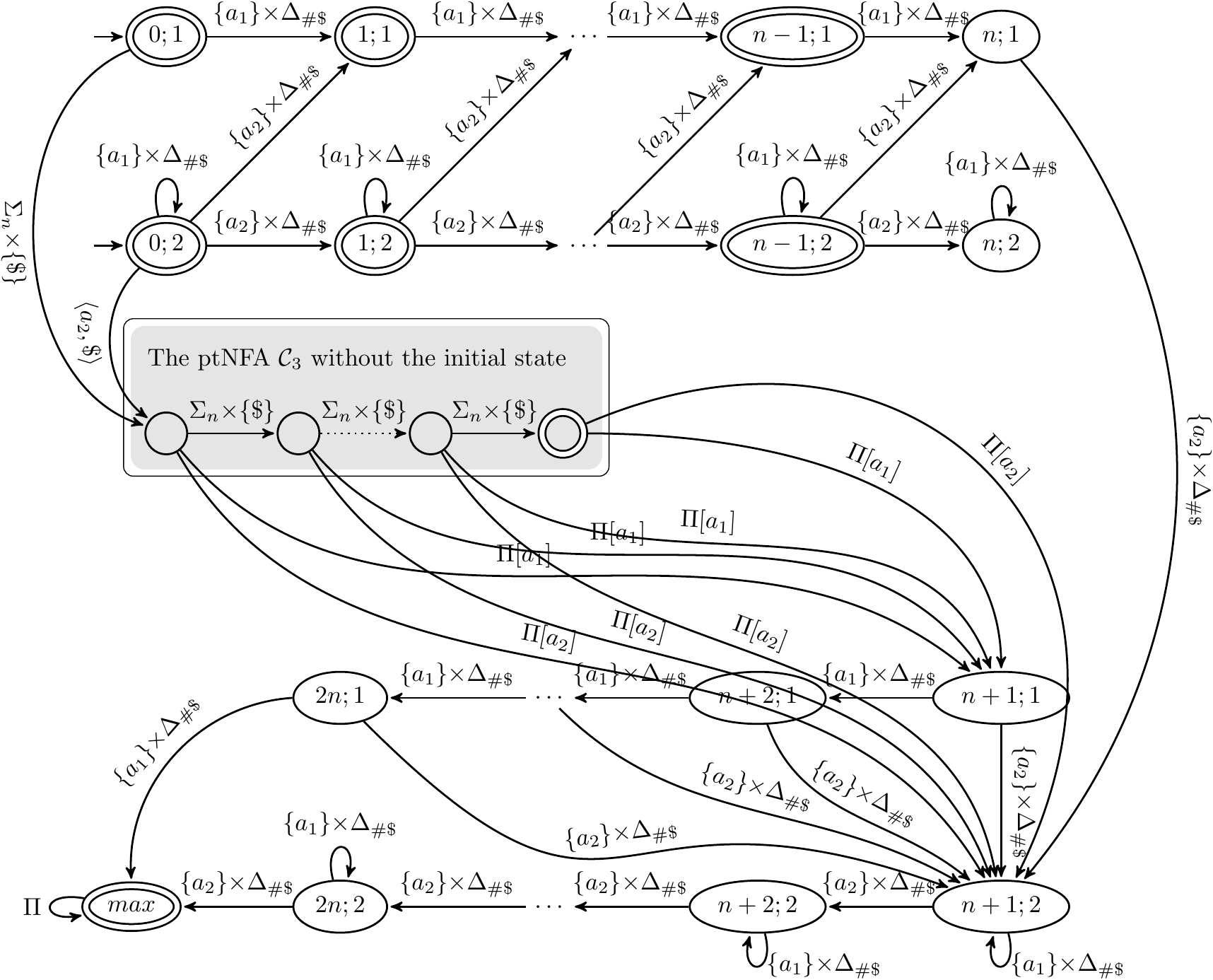}
    \caption{The ptNFA for {\bf (C.3)} illustrated for $n=2$; automaton $\mathcal{C}_3$, recognizing {\bf (C.3)} without the initial $\Pi^*$, is completed by adding transitions $q\xrightarrow{\Pi[a_i]} (n+1;i)$, where $\Pi[a_i]$ denotes all letters of $\Pi$ undefined in $q$ with the first component of the letter being $a_i$.}
    \label{figC3} 
  \end{figure}

\FloatBarrier

\section{Complexity of Deciding $k$-Piecewise Testability}
  The effort to simplify XML Schema using the BonXai language~\cite{MartensNNS15} led to the study of ($k$-)piecewise testable languages~\cite{icalp2013,HofmanM15}. A regular language over $\Sigma$ is {\em piecewise testable\/} if it is a finite boolean combination of languages of the form $\Sigma^* a_1 \Sigma^* a_2 \Sigma^* \cdots \Sigma^* a_n \Sigma^*$, where $a_i\in\Sigma$ for $i=1,\ldots,n$, $n\ge 0$. Let $k\ge 0$ be an integer. The language is {\em $k$-piecewise testable\/} if $n\le k$. The {\em $k$-piecewise testability problem\/} asks whether a given automaton recognizes a $k$-piecewise testable language.

  In this section, we study the complexity of deciding $k$-piecewise testability for partially ordered automata. Our results are summarized in Table~\ref{table1b}.
\begin{table*}\centering
  \ra{1.2}
  \begin{tabular}{@{}lllll@{}}\toprule
        & Unary alphabet
        & Fixed alphabet
        & \multicolumn{2}{l}{Arbitrary alphabet}\\
        \cmidrule{4-5}
        & $|\Sigma|=1$
        & $|\Sigma|\ge 2$
        & $k\le 3$
        & $k\ge 4$ \\
        \midrule
        DFA   & L-c            (Thm. \ref{DFAlc})
              & \NL-c          (Thm. \ref{kPTfixedDFA})
              & \NL-c          \cite{dlt15}
              & \coNP-c        \cite{KKP}    \\
        ptNFA & \NL-c          (Thm. \ref{thmP})
              & \coNP-c        (Thm. \ref{theorem27})
              & \multicolumn{2}{l}{\PSpace-c (Thm. \ref{thmMainUniv})} \\
       rpoNFA & \NL-c          
              & \coNP-c        \cite{mfcs16:mktmmt_full}
              & \multicolumn{2}{l}{\PSpace-c} \\
        poNFA & \NL-c          (Thm. \ref{poNFAnl})
              & \PSpace-c      (Thm. \ref{thm72})
              & \multicolumn{2}{l}{\PSpace-c} \\
          NFA & \coNP-c        (Thm. \ref{thm30})
              & \PSpace-c      \cite{dlt15}
              & \multicolumn{2}{l}{\PSpace-c \cite{dlt15}} \\
    \bottomrule
  \end{tabular}
  \caption{Complexity of deciding $k$-piecewise testability.}
  \label{table1b}
\end{table*}

  To simplify proofs, we make use of the following lemma that will save us a lot of work by directly obtaining the lower bounds from the complexity results for universality. 
  
  We first need some additional definitions. For $k\ge 0$, let $\sub_k(v) =\{u\in\Sigma^* \mid u\preccurlyeq v,\, |u|\le k\}$. For two words $w_1,w_2$, we write $w_1 \sim_k w_2$ if $\sub_k(w_1)=\sub_k(w_2)$. The relation $\sim_k$ is a congruence with finite index, and every $k$-piecewise testable language is a finite union of $\sim_k$ classes~\cite{Simon1972}. 
  \begin{lem}\label{lemma0k}
    Let $k\ge 0$ be a constant. Then the universality problem for ptNFAs (resp.\ rpoNFAs, poNFAs, NFAs, DFAs) is log-space reducible to the $k$-piecewise testability problem for ptNFAs (resp.\ rpoNFAs, poNFAs, NFAs, DFAs).
  \end{lem}
  \begin{proof}
    Let $\M$ over $\Sigma$ be a ptNFA (resp.\ rpoNFA, poNFA, NFA, DFA) recognizing a nonempty language. We construct a ptNFA (resp. rpoNFA, poNFA, NFA, DFA) $\M_{k}$ over $\Sigma$ from $\M$ as depicted in Figure~\ref{reduction2}. Namely, we add $|\Sigma|k$ new states $i_{j,1},\ldots,i_{j,|\Sigma|k}$ for every initial state $i_j$ of $\M$. For $1\le \ell < |\Sigma|k$, we add transitions from $i_{j,\ell}$ to $i_{j,\ell+1}$ and a transition from $i_{j,|\Sigma|k}$ to the initial state $i_j$ of $\M$ under all letters of $\Sigma$. The initial states of $\M_{k}$ are the states $i_{j,1}$, the accepting states are the accepting states of $\M$ and the states $i_{j,k+1},\ldots,i_{j,|\Sigma|k}$. Note that $\M_k$ is a ptNFA (resp. rpoNFA, poNFA, NFA, DFA) constructible in logarithmic space.
    \begin{figure}[b]
      \centering
      \includegraphics[scale=.9]{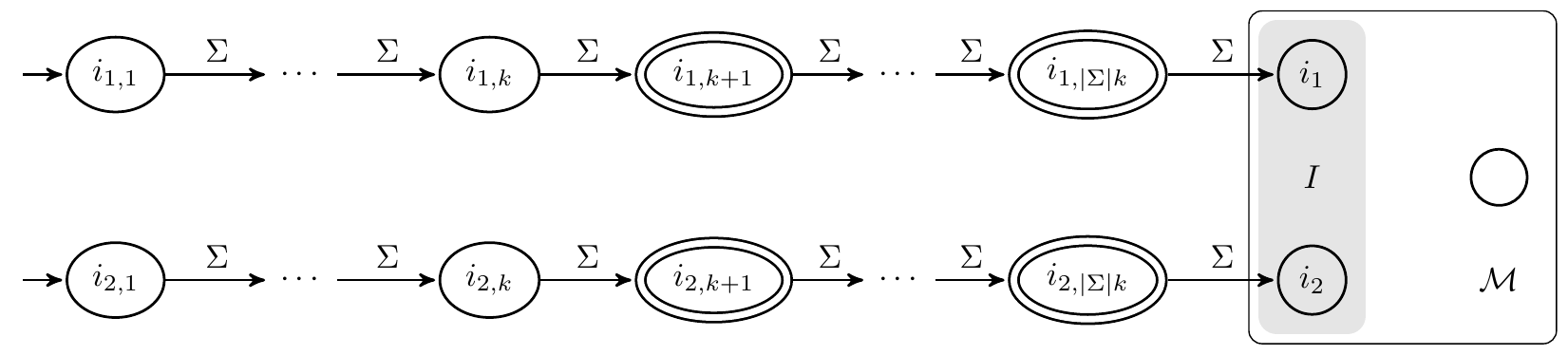}
      \caption{The ptNFA $\M_{k}$ constructed from a ptNFA $\M$ with two initial states.}
      \label{reduction2}
    \end{figure}

    If the language $L(\M)$ is universal, \ie, $L(\M)=\Sigma^*$, then the language $L(\M_k)=\Sigma^k\Sigma^*$ is $k$-piecewise testable because it consists of all words of length at least $k$, and hence if there are $x\in L(\M_k)$ and $y\notin L(\M_k)$, then the length of $y$ is less than $k$, and hence $x\sim_k y$ does not hold because $\sub_k(x)$ contains a word of length $k$ that is not in $\sub_k(y)$.
    
    If the language $L(\M)$ is not universal, then there exist $x \in L(\M)$ and $y\notin L(\M)$. Let $\Sigma$ be $\{a_1,a_2,\ldots,a_{|\Sigma|}\}$. Then $(a_1a_2\cdots a_{|\Sigma|})^kx \sim_k (a_1a_2\cdots a_{|\Sigma|})^k y$, since $\sub_{k}((a_1a_2\cdots a_{|\Sigma|})^k)=\{ u\in\Sigma^* \mid |u| \le k\}$, and $(a_1a_2\cdots a_{|\Sigma|})^kx\in L(\M_{k})$ and $(a_1a_2\cdots a_{|\Sigma|})^ky \notin L(\M_k)$, which shows that the language $L(\M_{k})$ is not $k$-piecewise testable.
  \end{proof}

  We immediately have the following consequences. The first consequence is that deciding $k$-piecewise testability for ptNFAs, where the alphabet may grow with the number of states, is \PSpace-complete.
  \begin{thm}\label{thmMainUniv}
    Deciding $k$-piecewise testability for ptNFAs is \PSpace-complete.
  \end{thm}
  \begin{proof}
    Membership follows from the results for NFAs, hardness follows from Lemma~\ref{lemma0k} and Theorem~\ref{thmMain}. 
  \end{proof}

  The second consequence is that the complexity decreases if we only consider ptNFAs over a fixed alphabet. We distinguish two cases: (i) at least binary alphabets, and (ii) unary alphabets. 
  \begin{thm}\label{theorem27}
    Let $\Sigma$ be a fixed alphabet with at least two letters. Deciding $k$-piecewise testability for ptNFAs over $\Sigma$ is \coNP-complete.
  \end{thm}
  \begin{proof}
    Hardness follows from Lemma~\ref{lemma0k} and Theorem~\ref{0ptNFAhard}, membership follows from the result for rpoNFAs~\cite[Corollary~24]{mfcs16:mktmmt_full}.
  \end{proof}

  This result is in contrast with an analogous result for DFAs. Deciding $k$-piecewise testability for DFAs over a fixed alphabet is in \Ptime~\cite{KKP}. A more precise complexity can be shown.
  \begin{thm}\label{kPTfixedDFA}
    Let $\Sigma$ be a fixed alphabet with at least two letters. Deciding $k$-piecewise testability for DFAs over $\Sigma$ is \NL-complete.
  \end{thm}
  \begin{proof}
    Hardness follows from Lemma~\ref{lemma0k} because deciding universality for DFAs is \NL-complete~\cite{Jones75}. Membership can be shown as follows. Since $\Sigma$ and $k$ are fixed, there is a constant number of $k$-piecewise testable languages over $\Sigma$, and hence we may assume that the minimal DFAs of all these languages are precomputed. Let $\A$ be a DFA. Then $L(\A)$ is $k$-piecewise testable if and only if it is equivalent to one of the precomputed languages. This can be verified in \NL by guessing a precomputed minimal DFA and checking equivalence (see the next section for more details).
  \end{proof}
  
  In comparison with ptNFAs or rpoNFAs, fixing the alphabet does not affect the complexity for poNFAs.
  \begin{thm}\label{thm72}
    Let $\Sigma$ be a fixed alphabet with at least two letters. Deciding $k$-piecewise testability for poNFAs over $\Sigma$ is \PSpace-complete.
  \end{thm}
  \begin{proof}
    Membership follows from the results for NFAs, hardness follows from Lemma~\ref{lemma0k} and the fact that deciding universality for poNFAs over $\Sigma$ is \PSpace-complete~\cite{mfcs16:mktmmt_full}.
  \end{proof}

  Theorems~\ref{theorem27} and~\ref{thm72} show hardness even for binary alphabets, which improves our recent result where the alphabet had at least three letters~\cite{ptnfas}. Furthermore, we point out that hardness in Theorem~\ref{theorem27} does not follow from the \coNP-hardness proof of Kl\'ima et al.~\cite{KKP} showing \coNP-completeness of deciding $k$-piecewise testability for DFAs for $k\ge 4$, since their proof requires a growing alphabet.

  It remains to consider the case of unary alphabets. 
  We first focus on the case of nondeterministic partially ordered automata and the variants thereof.
  \begin{thm}\label{poNFAnl}\label{thmP}
    Deciding $k$-piecewise testability for poNFAs, rpoNFAs, and ptNFAs over a unary alphabet is \NL-complete.
    It holds even if $k$ is given as part of the input.
  \end{thm}
  \begin{proof}
    Hardness follows from Lemma~\ref{lemma0k} and Theorem~\ref{thmMainNL}.
    We now show membership for poNFAs, which covers all the cases. Let $\A$ be a poNFA over the alphabet $\{a\}$ with $n$ states. If the language $L(\A)$ is infinite, then there exists $d\le n$ such that $a^da^* \subseteq L(\A)$; indeed, $L(\A)$ is infinite if and only if there is an accepting state that is reachable via a state with a self-loop, and hence $d$ is bounded by the number of states on such a path. Therefore, the language $L(\A)$ is {\em not\/} $k$-piecewise testable if and only if there exists $\ell$ with $k < \ell \le d$ such that $a^k \in L(\A)$ if and only if $a^\ell \notin L(\A)$. Such an $\ell$ can be guessed in binary and the property verified in \NL, since \NL is closed under complement. 
    If $L(\A)$ is finite, its complement, which is $k$-piecewise testable if and only if $L(\A)$ is, is infinite.
  \end{proof}
  
  Now we focus on the case of deterministic automata.
  \begin{thm}\label{DFAlc}
    Deciding $k$-piecewise testability for DFAs over a unary alphabet is L-complete.
  \end{thm}
  \begin{proof}
    Hardness follows from Lemma~\ref{lemma0k} and the fact that deciding universality for unary DFAs is \LogSpace-complete~\cite{Jones75}. 
    Membership in \LogSpace can be shown as follows. Let $n$ be the number of states of the DFA. Then the language is $k$-piecewise testable if and only if $a^k,a^{k+1},\ldots,a^{k+n}$ all belong to the language or none does. ($k+n$ because there may be a cycle to the initial state.)
  \end{proof}

  Finally, we focus on the case of general nondeterministic automata.
  \begin{thm}\label{thm30}
    Deciding $k$-piecewise testability for NFAs over a unary alphabet is \coNP-complete.
  \end{thm}
  \begin{proof}
    Hardness follows from Lemma~\ref{lemma0k} and the fact that deciding universality for unary NFAs is \coNP-complete~\cite{StockmeyerM73}. To show membership, we first show that deciding piecewise testability for NFAs over a unary alphabet is in \coNP. To do this, we show how to check non-piecewise testability in NP. Intuitively, we need to check that the corresponding DFA is partially ordered and confluent. However, confluence is trivially satisfied because there is no branching in a DFA over a single letter. Partial order is violated if and only if there exist three words $a^{\ell_1}$, $a^{\ell_2}$ and $a^{\ell_3}$ with $\ell_1 < \ell_2 < \ell_3$ such that $\delta(I,a^{\ell_1}) = \delta(I,a^{\ell_3}) \neq \delta(I,a^{\ell_2})$ and one of these sets is accepting (as a state of the DFA) and the other is not (otherwise they are equivalent). The lengths of the words are bounded by $2^n$, where $n$ denotes the number of states of the NFA, and can thus be guessed in binary. The matrix multiplication (fast exponentiation) can then be used to compute the sets of states reachable under those words in polynomial time.
    
    Thus, we can check in \coNP whether the language of an NFA is piecewise testable. If so, then it is $2^n$-piecewise testable, since the depth of the minimal DFA is bounded by $2^n$, where $n$ is the number of states of the NFA~\cite{ptnfas}. Let $M$ be the transition matrix of the NFA. To show that it is not $k$-piecewise testable, we need to find two $\sim_k$-equivalent words such that exactly one of them belongs to the language of the NFA. Since every $\sim_k$ class defined by $a^\ell$, for $\ell < k$, is a singleton, we need to find $k< \ell \le 2^n$ such that $a^k \sim_k a^\ell$ and only one of them belongs to the language. This can be done in nondeterministic polynomial time by guessing $\ell$ in binary, using the matrix multiplication to obtain the corresponding reachable sets in $M^k$ and $M^\ell$, and verifying that one set contains an accepting state and the other does not.
  \end{proof}

\section{Complexity of Deciding Piecewise Testability}
  The {\em piecewise testability problem\/} asks, given an automaton, whether it recognizes a piecewise testable language. We now study the complexity of deciding piecewise testability for partially ordered automata. Our results are summarized in Table~\ref{table2b}.
  \begin{table*}\centering
    \ra{1.2}
    \begin{tabular}{@{}llll@{}}\toprule
            & $|\Sigma|=1$
            & $|\Sigma|\ge 2$
            & $\Sigma$ is growing\\
          \midrule
            DFA    & L-c            (Thm. \ref{DFAlcb})
                    & \NL-c          \cite{ChoH91}\footnotemark{}
                    & \NL-c          \cite{ChoH91} \\
          rpoNFA    & $\checkmark$   (Thm. \ref{thm74})
                    & \coNP-c        (Thm. \ref{rpoNFApt})
                    & \PSpace-c      (Thm. \ref{rpoNFAtoPT})\\
          poNFA     & $\checkmark$   (Thm. \ref{thm74})
                    & \PSpace-c      (Thm. \ref{thm72b})
                    & \PSpace-c       \\
          NFA       & \coNP-c        (Thm. \ref{thm30b})
                    & \PSpace-c      \cite{tm2016}
                    & \PSpace-c      \cite{tm2016} \\
      \bottomrule
    \end{tabular}
    \caption{Complexity of deciding piecewise testability.}
    \label{table2b}
  \end{table*}
  \footnotetext{Cho and Huynh~\cite{ChoH91} showed hardness for a three-letter alphabet. However, the result holds also for binary alphabets, using, \eg, a reduction from the reachability problem for directed acyclic graphs with out-degree at most two.}

  To simplify proofs, we would like to use a result similar to Lemma~\ref{lemma0k}. Unfortunately, there is no such result preserving the alphabet. If there were, it would imply that deciding piecewise testability for ptNFAs has a nontrivial complexity, but these languages are trivially piecewise testable. Similarly, it would imply that deciding piecewise testability of unary (r)poNFAs is nontrivial, but we show below that they are trivially piecewise testable.
  
  Recall that $\R$-trivial languages, poDFA-languages, and rpoNFA-languages coincide.
  \begin{thm}\label{thm74}
    The classes of unary poNFA languages, unary $\R$-trivial languages, and unary piecewise testable languages coincide.
  \end{thm}
  \begin{proof}
    Since every piecewise testable language is an $\R$-trivial language, and every $\R$-trivial language is a poNFA language, we only need to prove that unary poNFA languages are piecewise testable. If the language of a poNFA is finite, then it is piecewise testable. If it is infinite, then there is an integer $n$ bounded by the number of states of the poNFA such that the poNFA accepts all words of length longer than $n$. The minimal DFA equivalent to the poNFA is thus partially ordered and confluent.
  \end{proof}

  We first discuss the complexity of deciding piecewise testability for unary DFAs.
  \begin{thm}\label{DFAlcb}
    Deciding piecewise testability for DFAs over a unary alphabet is L-complete.
  \end{thm}
  \begin{proof}
    To prove hardness, we reduce from the DAG-reachability problem where no vertex has more than one outgoing directed edge~\cite{Jones75}. 
    Let $G$ be a directed acyclic graph where no vertex has more than one outgoing directed edge with vertices $1,2,\ldots,n$, $n>1$. We define a DFA $\A=(\{0,1,\ldots,n\},\{a\},\delta,1,\{n\})$ using exactly the same reduction as Jones~\cite[Theorem~26]{Jones75} defining $\delta(i,a)=j$ if $(i,j)$ is an edge of $G$ and $i\neq n$, $\delta(n,a)=1$, and $\delta(i,a)=0$ for other values of $i$. Then $n$ is reachable from $1$ in $G$ if and only if the language $L(\A)$ is infinite as well as the language $\{a\}^* \setminus L(\A)$, which implies non-piecewise testability of the language $L(\A)$ because the minimal DFA for $L(\A)$ needs to have a nontrivial cycle. 
    
    To show membership, we need to check that there is no nontrivial cycle in the minimal DFA equivalent to the given DFA. This can be done by checking that the words $a^n,\ldots,a^{2n}$ all have the same accepting status, where $n$ is the number of states of the given DFA.
  \end{proof}

  Next we discuss the case of unary NFAs.
  \begin{thm}\label{thm30b}
    Deciding piecewise testability for NFAs over a unary alphabet is \coNP-complete.
  \end{thm}
  \begin{proof}
    Membership is shown in the proof of Theorem~\ref{thm30}.
    To show hardness, we modify the proof of Stockmeyer and Meyer~\cite{StockmeyerM73}.
    Let $\varphi$ be a formula in 3CNF with $n$ distinct variables, and let $C_k$ be the set of literals in the $k$th conjunct, $1 \le k \le m$. The assignment to the variables can be represented as a binary vector of length $n$. Let $p_1,p_2,\ldots,p_n$ be the first $n$ prime numbers. For a natural number $z$ congruent with 0 or 1 modulo $p_i$, for every $1\le i \le n$, we say that $z$ satisfies $\varphi$ if the assignment $(z \bmod p_1, z \bmod p_2,\ldots, z \bmod p_n)$ satisfies $\varphi$. Let 
    \[
      E_0 = \bigcup_{k=1}^{n} \bigcup_{j=2}^{p_k-1} 0^j\cdot (0^{p_k})^*
    \]
    that is, $L(E_0) = \{ 0^z \mid \exists k \le n, z \not\equiv 0 \bmod p_k \text{ and } z \not\equiv 1 \bmod p_k \}$ is the set of natural numbers that do not encode an assignment to the variables. For each conjunct $C_k$, we construct an expression $E_k$ such that if $0^z \in L(E_k)$ and $z$ is an assignment, then $z$ does not assign the value 1 to any literal in $C_k$. For example, if $C_k = \{x_{r}, \neg x_{s}, x_{t}\}$, for $1 \le  r,s,t \le n$ and $r,s,t$ distinct, let $z_k$ be the unique integer such that $0\le z_k < p_rp_sp_t$, $z_k \equiv 0 \bmod p_r$, $z_k \equiv 1 \bmod p_s$, and $z_k \equiv 0 \bmod p_t$. Then
    \[
      E_k = 0^{z_k} \cdot (0^{p_rp_sp_t})^*\,.
    \]
    Now, $\varphi$ is satisfiable if and only if there exists $z$ such that $z$ encodes an assignment to $\varphi$ and $0^z \notin L(E_k)$ for all $1\le k \le m$, which is if and only if $L(E_0 \cup \bigcup_{k=1}^{m} E_k) \neq 0^*$.
    
    The proof up to now shows that universality is \coNP-hard for NFAs over a unary alphabet. Let now $p_n\# = \Pi_{i=1}^{n} p_i$. If $z$ encodes an assignment of $\varphi$, then, for any natural number $c$, $z+c\cdot p_n\#$ also encodes an assignment of $\varphi$; indeed, if $z \equiv x_i \bmod p_i$, then $z + c\cdot p_n\# \equiv x_i \bmod p_i$, for every $1\le i\le n$. This shows that if, in addition,  $0^z \notin L(E_k)$ for all $k$, then $0^z (0^{p_n\#})^* \cap L(E_0 \cup \bigcup_{k=1}^{m} E_k) = \emptyset$. Since both the languages of the intersection are infinite, the minimal DFA recognizing the language $L(E_0 \cup \bigcup_{k=1}^{m} E_k)$ must have a non-trivial cycle alternating between accepting and non-accepting states. Therefore, if the language $L(E_0 \cup \bigcup_{k=1}^{m} E_k)$ is universal, then it is piecewise testable, and if it is non-universal, then it is not piecewise testable.
  \end{proof}

  We next show that deciding piecewise testability for poNFAs is \PSpace-complete even if the alphabet is binary.
  \begin{thm}\label{thm72b}
    Let $\Sigma$ be a fixed alphabet with at least two letters. Deciding piecewise testability for poNFAs over $\Sigma$ is \PSpace-complete.
  \end{thm}
  \begin{proof}
    Membership in \PSpace follows from the results for NFAs.
    \PSpace-hardness follows from an analogous result for rpoNFAs~\cite{mfcs16:mktmmt_full} where we construct, given a polynomial-space-bounded DTM $M$ and an input $x$, a binary poNFA $\A_x$ in polynomial time such that if $M$ does not accept $x$, then $L(\A_x)=\{0,1\}^*$, which is piecewise testable, and if $M$ accepts $x$, then $L(\A_x)$ is not $\R$-trivial, and hence neither piecewise testable. The language of $\A_x$ is thus piecewise testable if and only if $M$ does not accept $x$.
  \end{proof}  

  The case of rpoNFAs is more complicated. In the next theorem, we show that deciding piecewise testability for rpoNFAs is \PSpace-complete if the alphabet is not fixed, and then we discuss the case of rpoNFAs over a fixed (binary) alphabet.
  \begin{thm}\label{rpoNFAtoPT}
    Deciding piecewise testability for rpoNFAs is \PSpace-complete.
  \end{thm}
  \begin{proof}
    Membership follows from the result for NFAs. To prove hardness, we reduce the universality problem for rpoNFAs, which is \PSpace-complete in general and \coNP-complete for fixed alphabets~\cite{mfcs16:mktmmt_full}.
    
    Let $\A$ be an rpoNFA, and let $\Sigma$ be its alphabet. We construct an rpoNFA $\B$ from $\A$ by adding two fresh letters $a,b\notin \Sigma$ and by adding two new states $1$ and $2$. State $1$ is the only accepting state of $\B$. From every non-accepting state of $\A$, we add an $a$-transition to state $1$ and a $b$-transition to state $2$. From every accepting state of $\A$, we add an $a$- and a $b$-transition to $1$. Finally, states $1$ and $2$ contain self-loops under all lettrers from $\Sigma\cup\{a,b\}$. The construction is illustrated in Figure~\ref{proofOfThm20}. 
    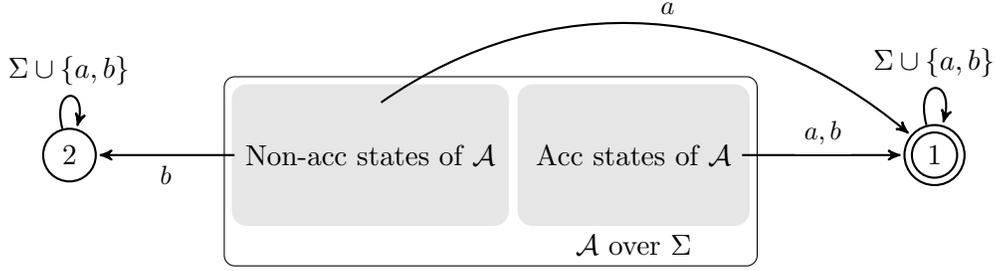
\begin{figure}
      \centering
      \begin{tikzpicture}[>=stealth',->,auto,shorten >=1pt,node distance=3cm,thick,
        state/.style={circle,minimum size=7mm,thick,draw=black,initial text=},
        eloop/.style={min distance=15mm,in=45,out=135,looseness=10},
        edgenode/.style={fill=white,font=\small},
        bigloop/.style={shift={(0,0.01)},text width=1.6cm,align=center}
        ]
        \tikzstyle{place}=[circle,thick,draw=black,minimum size=8mm]
        \tikzstyle{red place}=[place,draw=red!75,fill=red!20]
        \tikzstyle{transition}=[rectangle,thick,draw=black!75,fill=black!20,minimum size=4mm]

        \begin{scope}
          \node []      (a1)  {Acc states of $\A$};
          \node []      (d1) [above of=a1, node distance=.6cm]   {};
          \node []      (e1) [below of=a1, node distance=.6cm]   {};
          \node []      (d2) [left  of=d1, node distance=1.2cm]  {};
          \node []      (e2) [right of=e1, node distance=1.2cm]  {};
          \node []      (a2) [below of=e1, node distance=.6cm] {$\A$ over $\Sigma$};
          \node []      (d1) [left  of=a1, node distance=3.5cm] {Non-acc states of $\A$};
          \node []      (b1) [above of=d1, node distance=.6cm]  {};
          \node []      (c1) [below of=d1, node distance=.6cm]  {};
          \node []      (b2) [left  of=b1, node distance=1.5cm]  {};
          \node []      (c2) [right of=c1, node distance=1.5cm]  {};
          \node [state,accepting] (1)  [right of=a1, node distance=4cm] {$1$};
          \node [state]           (2)  [left  of=d1, node distance=4cm] {$2$};
        \end{scope}

        \path
          (a1) edge node[edgenode] {$a,b$} (1)
          (1) edge[loop above] node[bigloop] {$\Sigma\cup\{a,b\}$} (1)
          (d1) edge node[edgenode] {$b$} (2)
          (b1) edge[out=35,in=140] node[edgenode] {$a$} (1)
          (2) edge[loop above] node[bigloop] {$\Sigma\cup\{a,b\}$} (2)
          ;

        \begin{pgfonlayer}{background}
          \filldraw [line width=4mm,join=round,black!10]
            (c2.south  -| c2.east)  rectangle (b2.north  -| b2.west);
          \filldraw [line width=4mm,join=round,black!10]
            (e2.south  -| e2.east)  rectangle (d2.north  -| d2.west);
          \path (a2.south  -| b2.west)+(-0.3,0) node (a) {};
          \path (b1.north  -| e2.east)+(0.3,0.3)   node (b) {};
          \path[rounded corners, draw=black] (a) rectangle (b);
        \end{pgfonlayer}
      \end{tikzpicture}
      \caption{The illustration from the proof of Theorem~\ref{rpoNFAtoPT}.}
      \label{proofOfThm20}
    \end{figure}
    We now show that $L(\B)$ is piecewise testable if and only if $\A$ is universal.
    
    If $L(\A)=\Sigma^*$, then $L(\B)=\Sigma^*(a+b)(\Sigma\cup\{a,b\})^*$, because for every $w\in L(\A)$, the set of reachable states in $\B$ under $w$ contains an accepting state, and hence both $wa$ and $wb$ lead to state $1$. Language $\Sigma^*(a+b)(\Sigma\cup\{a,b\})^*$ is piecewise testable; it can be seen by computing the two-state minimal DFA and verifying that it is partially ordered and confluent.

    If $L(\A)$ is not universal, then there is a $w\in \Sigma^*\setminus L(\A)$. Then the set of states reachable under $w$ in $\B$ consists only of non-accepting states. By the construction, only state $1$ is reachable under $wa$, and only state $2$ is reachable under $wb$. Let $L_1 = wa(ba)^*$ and $L_2=wb(ab)^*$. Every word of $L_1$ is accepted by $\B$ whereas none word of $L_2$ is. If $L(\B)$ was piecewise testable, then there would be $k\ge 0$ such that for any words $w_1$ and $w_2$ with $w_1 \sim_k w_2$, either both words belong to $L(\B)$ or neither does. However, for every $k\ge 0$, we have that $wa(ba)^k \sim_k wb(ab)^k$ and the acceptance status of the two words is different. Therefore, the language $L(\B)$ is not piecewise testable.
  \end{proof}

  We now discuss the complexity of deciding piecewise testability for rpoNFAs over a fixed (binary) alphabet.
  \begin{thm}\label{rpoNFApt}
    Let $\Sigma$ be a fixed alphabet with at least two letters. Deciding piecewise testability for rpoNFAs over $\Sigma$ is \coNP-complete.
  \end{thm}
  \begin{proof}
    To show membership in \coNP, we proceed as follows. Let $\A=(Q,\Sigma,\delta_{\A},Q_0,F)$ be an rpoNFA over a fixed alphabet $\Sigma=\{a_1,a_2,\ldots,a_c\}$, and consider the minimal DFA $\D$ equivalent to $\A$. Then $L(\A)$ is piecewise testable if and only if $\D$ satisfies the UMS property (\cf\ Subsection~\ref{subsec3}). We proceed by first showing two auxiliary claims. 

    \begin{claim}\label{claim11.5}
      Every state of $\D$ is reachable by a word of polynomial length with respect to the size of $\A$.
    \end{claim}
    \begin{proof}
      To show this claim, we briefly recall basic definitions and results we need here. For more details, we refer the reader to Kr\"otzsch et al.~\cite{mfcs16:mktmmt_full}. 
      
      Similarly to piecewise testable languages, $\R$-trivial languages can be defined by a congruence $\sim^{\R}_{k}$ that considers subsequences of prefixes. For $x,y\in\Sigma^*$ and $k\ge 0$, we define $x \sim^{\R}_{k} y$ if and only if (i) for each prefix $u$ of $x$, there exists a prefix $v$ of $y$ such that $u \sim_k v$, and (ii) for each prefix $v$ of $y$, there exists a prefix $u$ of $x$ such that $u \sim_k v$. A regular language is \emph{$k$-$\R$-trivial} if it is a union of $\sim^{\R}_{k}$ classes, and it is $\R$-trivial if it is $k$-$\R$-trivial for some $k\ge 0$. Every $\sim^{\R}_{k}$ class has a unique minimal representative~\cite{BrzozowskiF80}. It is known that every $k$-$\R$-trivial language is also $(k+1)$-$\R$-trivial, and that the language recognized by a complete rpoNFA $\B$ is $\depth(\B)$-$\R$-trivial.

      Let $d$ be the depth of (the completion of) $\A$. Then, the language $L(\A)$ is $d$-$\R$-trivial~\cite[Theorem~8]{mfcs16:mktmmt_full}, and hence there is a congruence $\sim_d^{\R}$, where every $\sim_d^{\R}$ class has a unique minimal representative. Let $s$ be a state of $\D$ and $w$ a word reaching state $s$ from the initial state of $\D$. Let $w'$ denote the unique minimal representative of the $\sim_d^{\R}$-class containing $w$.
      If $w'$ leads $\D$ to a state $t\neq s$, then there is $u$ distinguishing $s$ and $t$ in $\D$ because $\D$ is minimal. Since $\sim_d^{\R}$ is a congruence and $w \sim_d^{\R} w'$, we have that $wu \sim_d^{\R} w'u$ and $wu\in L(\A)$ if and only if $w'u\notin L(\A)$, which is a contradiction to the fact that $L(\A)$ is $d$-$\R$-trivial, \ie, a union of $\sim^{\R}_{d}$ classes. Therefore, $w'$ leads $\D$ to state $s$. The length of $w'$ is polynomial, namely $O(d^c)$, where $c$ is the cardinality of $\Sigma$~\cite[Lemma~15]{mfcs16:mktmmt_full}.
    \end{proof}

    \begin{claim}\label{claim11.6}
      If $s$ is a state of $\D$ reachable by a word $w$, and $\{b_1,\ldots,b_m\}\subseteq \Sigma(s)$, then, for $v=b_1b_2\cdots b_m$, $\delta_{\A}(Q_0,w v^n)=\delta_{\A}(Q_0,w v^{n+1})$, where $n$ is the number of states of $\A$. Moreover, $\{b_1,\ldots,b_m\}\subseteq \Sigma(\delta_{\A}(Q_0,w v^n))\subseteq \Sigma(s)$.
    \end{claim}
    \begin{proof}
      Let $Q$ denote the set of states of $\A$ and extend the partial order of $\A$ to a linear order. Let $Q'\subseteq Q$ be the set of all states with self-loops under all letters of $\{b_1,\ldots,b_m\}$, and let $Q''=Q\setminus Q' =\{p_1,\ldots,p_{n'}\}$. We assume that $p_1 < p_2 < \ldots < p_{n'}$ in the linear order. Let $\delta_{\A}(Q_0,w) = X_1 \cup Z_1$, where $X_1 \subseteq Q''$ and $Z_1 \subseteq Q'$. Then $X_1 \subseteq \{p_i,p_{i+1},\ldots, p_{n'}\}$ for some $i$ such that $p_i \in X_1$. Let $X_1 \xrightarrow{v} X_2$. Then the minimal state $p_j$ of $X_2$ is strictly greater than $p_i$, since $\A$ is an rpoNFA, and hence $X_2 \subseteq \{p_{j},\ldots,p_{n'}\}$ with $j>i$. By induction, we have that $X_1 \xrightarrow{v^n} Z_2$, where $Z_2 \subseteq Q'$. Let $Z=Z_2\cup Z_1$. Then $\delta_{\A}(Q_0,w)=X_1\cup Z_1 \xrightarrow{v^n} Z =\delta_{\A}(Q_0,wv^n) \xrightarrow{v} Z=\delta_{\A}(Q_0,wv^{n+1})$.
      
      Since $\{b_1,\ldots,b_m\}\subseteq\Sigma(Z)$, it remains to show that $\Sigma(Z)\subseteq \Sigma(s)$. For the sake of contradiction, assume that there is $a \in \Sigma(Z)\setminus \Sigma(s)$. Then we have that $Z\xrightarrow{a} Z$, and hence, for any $u\in\Sigma^*$, $wv^nu$ belongs to $L(\A)$ if and only if $wv^nau$ does. However, in $\D$, $s\xrightarrow{a} s'$ for some $s'\neq s$, and hence, since $v^n\in\Sigma(s)^*$, there is $u\in\Sigma^*$ such that $wv^nu$ belongs to $L(\D)=L(\A)$ if and only if $wv^nau$ does not; a contradiction.
    \end{proof}

    Now, $L(\A)=L(\D)$ is not piecewise testable if and only if there are states $s\neq t$ in $\D$ such that $s$ and $t$ are two maximal states of the connected component of $G(\D,\Sigma(s))$ containing $s$; that is, $\Sigma(s)\subseteq\Sigma(t)$. By Claim~\ref{claim11.5}, there are two words $w_s$ and $w_t$ of polynomial length with respect to the size of $\A$ reaching the states $s$ and $t$ of $\D$, respectively. Then, in $\A$, $Q_0\xrightarrow{w_s} S$ for some $S\subseteq Q$. If $\Sigma(s)=\{b_1,\ldots,b_{c'}\}$, let $v=b_1\cdots b_{c'}$. Then, by Claim~\ref{claim11.6}, $S \xrightarrow{v^n} X_s \xrightarrow{v} X_s$, where $n$ is the number of states of $\A$, and $\Sigma(X_s)= \Sigma(s)$. Analogously, $Q_0\xrightarrow{w_t} T \xrightarrow{v^n} X_t \xrightarrow{v} X_t$ with $\Sigma(X_s)\subseteq\Sigma(X_t)\subseteq\Sigma(t)$.
    Furthermore, since the length of $v^n$ is $nc'$, which is polynomial in the size of $\A$, the length of the two words $w_1 = w_s v^n$ and $w_2 = w_t v^n$ is polynomial in the size of $\A$.
   
    Altogether, we have shown that the language of $\A$ is not piecewise testable if and only if there are two different words $w_1$ and $w_2$ of polynomial length in the size of $\A$ such that 
    \begin{itemize}
      \item $Q_0 \xrightarrow{w_1} X_s$ and $Q_0 \xrightarrow{w_2} X_t$, 
      \item $X_s$ and $X_t$ are maximal with respect to $\Sigma(X_s)$, and
      \item $X_s$ and $X_t$ are non-equivalent as states of the subset automaton -- which can be checked by guessing a word that distinguishes them; by Claim~\ref{claim11.5} applied to $\A$ with the set of initial states $X_s$ (resp. $X_t$) instead of $Q_0$, which results in a subautomaton of $\D$, and the existence of unique minimal representatives of the equivalence classes, \cf\ the proof of the claim, such a word is of polynomial length.
    \end{itemize}
    This shows that non-piecewise testability of an rpoNFA-language over a fixed alphabet is in \NP, which was to be shown.

    To show hardness, we reduce the DNF validity. Let $U=\{x_1,\ldots,x_n\}$ be a set of variables and $\varphi = \varphi_1 \lor \ldots \lor \varphi_m$ be a formula in DNF, where every $\varphi_i$ is a conjunction of literals. We assume that no $\varphi_i$ contains both $x$ and $\neg x$. For every $i=1,\ldots,m$, we define $\beta_i = \beta_{i,1}\beta_{i,2}\ldots\beta_{i,n}$, where 
    \[
      \beta_{i,j} = \left\{
        \begin{array}{ll}
          0+1 & \text{ if neither } x_j \text{ nor } \neg x_j \text{ appear in } \varphi_i\\
          0   & \text{ if } \neg x_j \text{ appears in } \varphi_i\\
          1   & \text{ if } x_j \text{ appears in } \varphi_i
        \end{array}
        \right.
    \]
    for $j=1,2,\ldots,n$. Let $\beta = \sum_{i=1}^{m} \beta_{i}$. Then $w\in L(\beta)$ if and only if $w$ satisfies some $\varphi_i$, that is, $L(\beta) = \{0,1\}^n$ if and only if $\varphi$ is valid.

    We construct an rpoNFA $\M$ as follows. For every $\beta_{i}$, we construct a deterministic path $q_{i,0} \xrightarrow{\beta_{i,1}} q_{i,1} \xrightarrow{\beta_{i,2}} q_{i,2} \ldots \xrightarrow{\beta_{i,n}} q_{i,n} \xrightarrow{0,1} q_{i,n}$ with a self-loop at the end; if $\beta_{i,k}=0+1$, the notation means that there are two transitions under both letters $0$ and $1$. Then we add a path $\alpha_1 \xrightarrow{0,1} \alpha_2 \xrightarrow{0,1} \ldots \xrightarrow{0,1} \alpha_{n}$ to accept all words of length less than $n$. The automaton $\M$ consists of these paths, where the initial states are $\{q_{i,0} \mid i=1,\ldots,m\}\cup\{\alpha_1\}$ and the accepting states are $\{q_{i,n} \mid i=1,\ldots,m\}\cup\{\alpha_1,\ldots,\alpha_{n}\}$. Notice that $\M$ is an rpoNFA accepting the language $L(\M) = L(\beta)\{0,1\}^* \cup \{w \in \{0,1\}^* \mid |w| < n\}$.
    
    If $L(\beta) = \{0,1\}^n$, then $L(\M)=\{0,1\}^*$ is piecewise testable.
    
    If $L(\beta) \neq \{0,1\}^n$, we show that $L(\M)$ is not piecewise testable using the UMS property on the minimal DFA equivalent to $\M$. Since $\M$ is an rpoNFA, the minimal DFA is partially ordered. By the assumption that $L(\beta) \neq \{0,1\}^n$, there is a $w\in\{0,1\}^n$ such that $w\{0,1\}^* \cap L(\M) = \emptyset$. Since $L(\beta)\neq\emptyset$ by the construction, there is $w'\in L(\beta)$, which implies that $w'\{0,1\}^* \subseteq L(\M)$. Since no word of $w\{0,1\}^*$ is accepted by $\M$, there is a path from the initial state of the minimal DFA to a rejecting state, say $q_r$, that is maximal under $\{0,1\}$. Similarly, since all words of $w'\{0,1\}^*$ are accepted by $\M$, there is a path in the minimal DFA to an accepting state, say $q_a$, that is maximal with respect to $\{0,1\}$. But then $q_a$ and $q_r$ are two maximal states violating the UMS property of the minimal DFA.

    Thus, $L(\M)$ is piecewise testable if and only if $\varphi$ is valid.
  \end{proof}

\section{Inclusion and Equivalence}\label{IandE}
  A consequence of the complexity of universality is the worst-case lower-bound complexity for the inclusion and equivalence problems. These problems are of interest, \eg, in optimization. The problems ask, given languages $K$ and $L$, whether $K\subseteq L$, resp. $K=L$.  Although equivalence means two inclusions, complexities of these two problems may differ significantly, \eg, inclusion is undecidable for deterministic context-free languages~\cite{Friedman76} while equivalence is decidable~\cite{Senizergues97}. 
  
  Since universality can be expressed as the inclusion $\Sigma^* \subseteq L$ or the equivalence $\Sigma^* = L$, we immediately obtain the hardness results for inclusion and equivalence from the results for universality. Therefore, it remains to show memberships of our results summarized in Tables~\ref{table3b} and~\ref{table4b}.
  
  %inclusion
\begin{table*}\centering
  \ra{1.2}
  \begin{tabular}{@{}lllll@{}}\toprule
                    & \multicolumn{4}{c}{$B$} \\
    \cmidrule{2-5}
    $A$ & DFA & ptNFA \& rpoNFA & poNFA & NFA\\ \midrule
    DFA   & \LogSpace/\NL
          & \NL/\coNP/\PSpace
          & \NL/\PSpace
          & \coNP/\PSpace \\
    ptNFA & \NL
          & \NL/\coNP/\PSpace
          & \NL/\PSpace
          & \coNP/\PSpace \\
    rpoNFA& \NL
          & \NL/\coNP/\PSpace
          & \NL/\PSpace
          & \coNP/\PSpace \\
    poNFA & \NL
          & \NL/\coNP/\PSpace
          & \NL/\PSpace
          & \coNP/\PSpace \\
    NFA   & \NL
          & \NL/\coNP/\PSpace
          & \NL/\PSpace
          & \coNP/\PSpace\\
    \bottomrule
  \end{tabular}
  \caption{Complexity of deciding inclusion $L(A)\subseteq L(B)$ (unary/fixed[/growing] alphabet), all results are complete for the given class.}
  \label{table3b}
\end{table*}

\begin{table*}\centering
  \ra{1.2}
  \begin{tabular}{@{}lllll@{}}
        \toprule
                  & DFA & ptNFA \& rpoNFA & poNFA & NFA \\ \midrule
           DFA    & \LogSpace/\NL
                  & \NL/\coNP/\PSpace
                  & \NL/\PSpace
                  & \coNP/\PSpace \\
        ptNFA     & 
                  & \NL/\coNP/\PSpace
                  & \NL/\PSpace
                  & \coNP/\PSpace \\
        rpoNFA    & 
                  & \NL/\coNP/\PSpace
                  & \NL/\PSpace
                  & \coNP/\PSpace \\
        poNFA     & 
                  & 
                  & \NL/\PSpace
                  & \coNP/\PSpace \\
        NFA       & 
                  & 
                  & 
                  & \coNP/\PSpace \\
    \bottomrule
  \end{tabular}
  \caption{Complexity of deciding equivalence (unary/fixed[/growing] alphabet), the problems are complete for the given classes.}
  \label{table4b}
\end{table*}

  \subsection{Proofs}
  Let $A$ be an automaton of any of the considered types. We now discuss the cases depending on the type of $B$. We assume that both automata are over the same alphabet specified by $B$.
  
  If $B$ is a DFA, then $L(A)\subseteq L(B)$ if and only if $L(A) \cap L(\overline{B}) = \emptyset$, which can be checked in \NL (or in \LogSpace if both automata are unary DFAs), where $\overline{B}$ denotes the DFA obtained by complementing $B$. This covers the first column of Table~\ref{table3b}.
  
  If $B$ is an rpoNFA over a fixed alphabet, then deciding $L(A)\subseteq L(B)$ is in \coNP~\cite[Theorem~23]{mfcs16:mktmmt_full}. Furthermore, the case of a unary alphabet follows from the case of unary poNFAs, and the case of a growing alphabet from the case of general NFAs discussed below.
  
  If $B$ is a unary poNFA, we distinguish several cases. First, deciding whether the language of an NFA is finite is in \NL. Thus, if $L(A)$ is infinite and $L(B)$ finite, the inclusion does not hold. If both the languages are finite, then the number of words is bounded by the number of states, and hence the inclusion can be decided in \NL. If $L(B)$ is infinite, then there is $n$ bounded by the number of states of $B$ such that $L(B)$ contains all words of length at least $n$. Thus, the inclusion does not hold if and only if there is a word of length at most $n$ in $L(A)$ that is not in $L(B)$, which can again be checked in \NL.
  
  If $B$ is an NFA, then deciding $L(A)\subseteq L(B)$ is in \PSpace using the standard on-the-fly computation of $\overline{B}$ and deciding $L(A) \cap L(\overline{B}) = \emptyset$.
    
  If $B$ is a unary NFA, then if $L(B)$ is finite, we proceed as in the case of $B$ being a unary poNFA. Therefore, assume that $L(B)$ is infinite and $B$ has $n$ states. Then the minimal DFA recognizing $L(B)$ has at most $2^n$ states (a better bound is given by Chrobak~\cite{chrobak}). If the inclusion $L(A)\subseteq L(B)$ does not hold and $A$ has $m$ states, then there exists $k \le m\cdot 2^n$, the number of states of $A\times \overline{B}$, such that $a^k \in L(A)\setminus L(B)$. We can guess $k$ in binary and verify that the inclusion does not hold in polynomial time by computing the reachable states under $a^k$ using the matrix multiplication. Hence, checking that the inclusion holds is in \coNP.
  
  %equivalence
  Notice that the upper-bound complexity for equivalence follows immediately from the upper-bound complexity for inclusion, which completes this section.

\section{Conclusion}
  We studied the complexity of deciding universality for ptNFAs, a type of nondeterministic finite automata the expressivity of which coincides with level 1 of the Straubing-Th\'erien hierarchy. Our proof showing \PSpace-completeness required a novel and nontrivial extension of our recent construction for self-loop-deterministic poNFAs. Consequently, we obtained \PSpace-completeness for several restricted types of poNFAs for problems including inclusion, equivalence, and ($k$-)piecewise testability.

\subsection*{Acknowledgements.}
  We thank O. Kl\'ima, M. Kunc and L. Pol\'ak for providing us with their manuscript~\cite{KKP}, and we gratefully acknowledge very useful suggestions and comments of the anonymous referees.

\bibliographystyle{plain}
\bibliography{mybib}

\end{document}